\documentclass{sig-alternate-10pt}
\usepackage{cite,subfigure,times,url}
\usepackage[ruled,vlined]{algorithm2e}
\usepackage{multirow}
\usepackage[T1]{fontenc}
\usepackage{textcomp}
\usepackage[top=0.75in,bottom=1.00in,left=0.625in,right=0.625in,letterpaper]{geometry}

\newtheorem{theorem}{\textbf{Theorem}}
\newtheorem{corollary}{\textbf{Corollary}}

\begin{document}

\title{Your Actions Tell Where You Are: Uncovering Twitter Users in a Metropolitan Area}


\numberofauthors{3}
\author{
\alignauthor Jinxue Zhang \\
\affaddr{Arizona State University} \\
\email{jxzhang@asu.edu}
\alignauthor Jingchao Sun \\
\affaddr{Arizona State University} \\
\email{jcsun@asu.edu}
\and
\alignauthor Rui Zhang  \\
\affaddr{University of Hawaii} \\
\email{ruizhang@hawaii.edu}
\alignauthor Yanchao Zhang \\
\affaddr{Arizona State University} \\
\email{yczhang@asu.edu}
}

\maketitle

\begin{abstract}

Twitter is an extremely popular social networking platform. Most Twitter users do not disclose their locations due to privacy concerns. Although inferring the location of an individual Twitter user has been extensively studied, it is still missing to effectively find the majority of the users in a specific geographical area without scanning the whole Twittersphere, and obtaining these users will result in both positive and negative significance. In this paper, we propose LocInfer, a novel and lightweight system to tackle this problem. LocInfer explores the fact that user communications in Twitter exhibit strong geographic locality, which we validate through large-scale datasets. Based on the experiments from four representative metropolitan areas in U.S., LocInfer can discover on average 86.6\% of the users with 73.2\% accuracy in each area by only checking a small set of candidate users. We also present a countermeasure to the users highly sensitive to location privacy and show its efficacy by simulations.

\end{abstract}

\section{Introduction}\label{sec:Intro}

Twitter is an extremely popular social networking tool for communicating through short messages called tweets. As of July 2014, Twitter has 255 million monthly active users and 500 million daily tweets.  Due to such massive user bases and popular usage, Twitter has been increasingly used in social communications, information campaigns, public relations, political campaigns, pandemic and crisis situations, marketing, and many other public/private contexts.

User privacy is arguably a major concern about Twitter. Specifically, user profiles and tweets may contain sensitive information about life, work, health, hobbies, political opinions, etc. Twitter currently offers little protection for user profiles and tweets which are virtually visible to anyone with or without an account.\footnote{Although Twitter allows a user to make his information visible to approved followers only, this privacy enhancement is rarely used in practice.} Consequently, many users employ pseudonyms instead of real names in their profiles. In addition, Twitter users often hide their home locations (location for short thereafter), which are \emph{permanent and static city-level} regions (e.g., Philadelphia) where most of their daily activities occur. Specifically, they may either not indicate their locations or report very general locations (e.g., state-level) in their profiles; they may not indicate their locations in their tweets either. For example, less than 34\% of Twitter users explicitly specify their locations in their profiles  \cite{LeetaMap13}, only 16\% of Twitter users indicate city-level locations, and only 0.5\% of tweets have a geo-tag \cite{LiTow12}.

There have been some efforts to infer a Twitter user's hidden location. Content-based methods \cite{HechtTwe11, ChengYou10, MahmuWhe12, LiTow12, MahmuHom14} try to infer  hidden locations based on geographic hints such as city landmarks in tweets. For example, a user who frequently mentions ``Golden Bridge'' in his tweets may indicate his location in the Bay Area. In contrast, network-based methods \cite{BacksFin10, McGeeLoc13, JurgeTha13, YamagLan13, ComptGeo14} leverage the fact that geographically-close people tend to form a connection or community in Online Social Networks (OSNs) \cite{QuercSoc12}, so a user's location can be inferred from those of his online neighbors (or neighbors' neighbors, etc). Based on different estimation techniques, all these efforts \cite{HechtTwe11, ChengYou10, MahmuWhe12, LiTow12, MahmuHom14, BacksFin10, McGeeLoc13, JurgeTha13, YamagLan13, ComptGeo14} seek to address the same question: how can we infer a Twitter user's hidden location from all his location-related tweets and/or OSN neighbors' locations?

This paper targets a different and more challenging problem: is it feasible to efficiently discover the majority of Twitter users in any city-level metropolitan area ($\mathbb{A}$) without collaborating with Twitter? Since only 16\% of Twitter users register city-level locations\cite{LiTow12}, it is infeasible to tackle our problem by directly checking users' tweets and profiles. In addition, directly applying any prior solution \cite{HechtTwe11, ChengYou10, MahmuWhe12, LiTow12, MahmuHom14, BacksFin10, McGeeLoc13, JurgeTha13, YamagLan13, ComptGeo14} would inevitably involve checking every (255 million) Twitter user's tweets, followers, and/or followees, thus leading to a prohibitive cost.


An affirmative answer to our target problem above would have significant \emph{positive} and \emph{negative} impacts. On the \emph{positive} side, finding the majority of the users in a specific area can not only benefit many applications such as local event detection and recommendation, business marketing, and em- \newline ergency-alert dissemination, but also offer a feasible way to sample Twitter to facilitate the research concerning geographically related information. On the \emph{negative} side, if an attacker can infer the majority of the Twitter users in a specific area, he could easily combine the location information with user tweets to better profile Twitter users who may or may not use pseudonyms, thus breaching their privacy and subjecting them to many identity-based attacks. Moreover, the Twitter users with exposed locations are vulnerable to large-scale location-based or geo-targeted spam campaigns \cite{SridhTwi12}.


In this paper, we propose LocInfer, a novel and lightweight solution to the above problem for the first time in literature. The design of LocInfer is driven by two conjectures. First, a small but nontrivial fraction of users (15.9\% on average in our datasets) have specified a credible location in the target area $\mathbb{A}$ in their personal profiles, each of which is referred to as a \emph{seed user} hereafter. Second, user communications in Twitter exhibit strong geographic locality in the sense that the users in the same area tend to interact more often than with those from outside. We confirm these two conjectures through large-scale datasets involving four representative metropolitan areas in U.S. Built upon these conjectures, LocInfer iteratively checks the immediate neighbors of the seed set, and the users who have tight connections with the seed set become new seeds and are added to the seed set. The final seed set contains the majority of Twitter users in $\mathbb{A}$ with overwhelming probability. LocInfer is highly efficient because only a small number of candidate users need to be checked in contrast to almost all the Twitter users when the existing methods \cite{HechtTwe11, ChengYou10, MahmuWhe12, LiTow12, MahmuHom14, BacksFin10, McGeeLoc13, JurgeTha13, YamagLan13, ComptGeo14} are applied to our problem.

Our contributions can be summarized as follows.
\begin{itemize}
\item We motivate and formulate the problem of large-scale location inference, which is challenging given that only a small fraction of Twitter users have specified a credible city-level location in their personal profiles.

\item We design LocInfer, a novel and lightweight solution that can uncover the majority of the Twitter users in a specific metropolitan area.

\item We conduct extensive experiments to evaluate LocInfer using four large-scale datasets. Our results show that LocInfer can successfully discover on average 86.6\% of the users with 73.2\% accuracy.

\item We propose a countermeasure against LocInfer for the Twitter users worrying about their location privacy and evaluate its effectiveness via simluations.
\end{itemize}

The rest of this paper is organized as follows. Section ~\ref{sec:DLA} defines the problem. Section~\ref{sec:Validation} validates our two conjectures through four large-scale datasets. Section~\ref{sec:LocInfer} details the LocInfer design. Section~\ref{sec:Evaluate} evaluates LocInfer and our countermeasure. Section~\ref{sec:RW} surveys the related work. Section~\ref{sec:CFW} concludes this paper and presents some future work.

\section{Problem Statement, Terms and Notation}\label{sec:DLA}

We use a directed and weighted multigraph\footnote{In a multigraph, two vertices may be connected by more than one edge.} to model the diverse communications between Twitter users. In Twitter, people can follow others without mutual consent; they can mention others in their own tweets; they can also reply to or retweet others' tweets. We classify these communications into two categories: \emph{following} and \emph{interacting} (retweeting, replying, and mentioning), denoted by symbols $\mathcal{F}$ and $\mathcal{I}$, respectively. Such diverse communications are modeled as a directed and weighted multigraph $G=\langle V, E\rangle$, where each vertex $v\in V$ represents a user. We refer to a directed edge for the following type as a following edge and a directed edge for the interacting type as an interacting edge. A following edge $e^\mathcal{F}_{ij} \in E$ is formed when user $i$ followed $j$; we call user $i$ a \emph{follower} of $j$ and $j$ a \emph{followee} of $i$. In contrast, an interacting edge $e^\mathcal{I}_{ij} \in E$ is formed when user $i$ mentioned, replied to, or retweeted $j$ at least once; we call user $i$ a \emph{responder} of $j$ and $j$ an \emph{initiator} of $i$. To model the interaction strength, we define $w(e^\mathcal{I}_{ij})$, the weight of edge $e^\mathcal{I}_{ij}$, as the total number of retweets, replies, and mentions from user $i$ to $j$. For consistency, we also define the weight of any following edge as one.  We use $N^{\mathcal{F}}_I(u), N^{\mathcal{F}}_O(u), N^{\mathcal{I}}_I(u), N^{\mathcal{I}}_O(u)$ to represent $u$'s one-hop followers, followees, responders, and initiators, respectively. We also define the one-hop neighbors of $u$ as $N(u)=N^{\mathcal{F}}_I(u)\cup N^{\mathcal{F}}_O(u)\cup N^{\mathcal{I}}_I(u)\cup N^{\mathcal{I}}_O(u)$.\\

\noindent\textbf{Large-Scale Location Inference.} Given a Twitter multigraph $G=\langle V, E\rangle$ and a target metropolitan area $\mathbb{A}$, we aim to obtain a target user list $U$ which contains the majority of Twitter users in $\mathbb{A}$ without collaborating with Twitter.\\

\noindent\textbf{Design goals.} LocInfer is designed with the following goals.
\begin{itemize}
\item \emph{High coverage}. The target user list $U$ should cover the majority of Twitter users in $\mathbb{A}$. If we denote the actual Twitter users in $\mathbb{A}$ by $U^*$, the coverage can be computed as $|U \cap U^*| / |U^*|$.

\item \emph{High accuracy\footnote{Note that \emph{coverage} and \emph{accuracy} correspond to the widely-used \emph{recall} and \emph{precision}, respectively. In this paper we use the coverage and accuracy to make the meaning more straightforward in the context of user uncovering in an area.}}. The target users in $U$ should be indeed located in $\mathbb{A}$. The accuracy can be computed as $|U \cap U^*| / |U|$.

\item \emph{Efficiency}. LocInfer should only involve checking Twitter users proportional in quantity to the population in $\mathbb{A}$ in contrast to existing methods \cite{HechtTwe11, ChengYou10, MahmuWhe12, LiTow12, MahmuHom14, BacksFin10, McGeeLoc13, JurgeTha13, YamagLan13, ComptGeo14} which all need to check all the Twitter users. This efficiency requirement is particularly important because without Twitter's collaboration, the only free way to obtain the users' information is via third-party APIs, which is time-consum- \newline ing as Twitter has strict rate limits on APIs invoking \cite{api13}. For example, an authenticated user can only invoke the get-followers API 15 times per 15 minutes. Hence if we invoke this API once for each of the 255 million Twitter users, it will spend a single authenticated user about 485 years to obtain all the Twitter users' followers.
\end{itemize}

\begin{table*}[t]
    \caption{Seed users in four metropolitan areas in U.S.}
    \centering
    \begin{tabular}{|c|c|c|c|c|}
       \hline
        \multirow{2}{*}{Area $\mathbb{A}$} & Population & \multirow{2}{*}{\#Twitter users} & \#seed users & \#seed users with \\
                        & (rank in U.S.)  &   &   (over \#Twitter users)  &$\geq$ 1 million followers \\
       \hline
        Tucson (\texttt{TS}) &  996,544 (57th) & 150,478 & 28,161 (18.65\%) & 0\\
        Philadelphia (\texttt{PI}) &  6,034,678 (7th) & 911,236 & 144,033 (15.9\%) & 3\\
        Chicago (\texttt{CI}) &  9,522,434 (3rd) & 1,437,888 & 318,632 (22.21\%) & 11\\
        Los Angeles (\texttt{LA}) &16,400,000 (2nd) &  2,476,400 & 300,148 (12.12\%) & 174\\
        \hline
    \end{tabular}
    \label{tlb:datasets}
\end{table*}

\section{Conjectures Validation}\label{sec:Validation}

As we mentioned in Section~\ref{sec:Intro}, LocInfer is built upon two important conjectures.
\begin{itemize}
\item \emph{Conjecture~1}: A small but nontrivial fraction of users have specified a credible location in the target area $\mathbb{A}$ in their personal profiles.
\item \emph{Conjecture~2}: User communications in Twitter exhibit strong geographic locality in the sense that the users in the same area tend to communicate more often than with those from outside.
\end{itemize}
In this section, we validate these two conjectures using four large-scale datasets.

\subsection{Data Collection}

We collect ground-truth Twitter users in different metropolitan areas by checking the self-reported locations in their profiles, a methodology that has been used to obtain the ground truth in \cite{HechtTwe11, ChengYou10, MahmuWhe12, LiTow12, MahmuHom14, BacksFin10, McGeeLoc13, JurgeTha13, YamagLan13, ComptGeo14}. Specifically, we use the Twitter geo-search API designed to return the recent or popular tweets in a specified \emph{geo-circle} defined by latitude, longitude, and radius \cite{api13}. For any interested area $\mathbb{A}$, we convert it into a geo-circle for the geo-search API, and we do not differentiate $\mathbb{A}$ and its corresponding geo-circle hereafter. The geo-search API returns the tweets from three types of users.
\begin{itemize}
\item \emph{Geo-tagged users}: The users who recently published some tweets with a geo-tag in $\mathbb{A}$.

\item \emph{Geo-profiled users}: The users whose personal profiles containing a location in $\mathbb{A}$.

\item \emph{Retweeting users}: The users who recently retweeted some geo-tagged or geo-profiled users' tweets in $\mathbb{A}$.
\end{itemize}
Among them, we only use the geo-profiled users to build our datasets, because retweeting users are likely not in $\mathbb{A}$, and geo-tagged users may have just traveled to some places within the geo-circle instead of living there. Moreover, since the result of each geo-search API invoking corresponds to a random sampling of the active Twitter users, we keep invoking the geo-search API until no significantly more geo-profiled users can be discovered.

The self-reported locations have been found reliable \cite{ComptGeo14}, but the results from the geo-search API are still noisy for two reasons. First, the location descriptions in many users' profiles are ambiguous and arbitrary. For example, people living in Los Angeles may specify their locations as ``South California'', or ``Los Angeles'', or ``LA'', or just ``CA.'' Second, the geo-search API often needs to covert a location description into a longitude-latitude pair for comparison with the specified geo-circle. Such conversions are often problematic and thus lead to wrong results. For example, when we searched the users in San Francisco Bay Area, the geo-search API returned some users in other places or even nonsense descriptions such as ``somewhere you're not'' and ``wherever you not.''

We thus refine the geo-profiled users as follows. For each user, we further verify whether his/her location description indeed contains a city name in $\mathbb{A}$. For this purpose, we first obtain the list of city names in $\mathbb{A}$ from the latest U.S. gazetteer data \cite{gazetteer13} and then compare the location description with the list. If there is an intersection, the user is considered a ground-truth user in $\mathbb{A}$.

\subsection{Datasets}

Using the above method, we collect user data in four met- \newline ropolitan areas of Tuscon (Arizona), Philadelphia, Chicago, and Los Angeles. Our data collection ran from January to June 2014. Table~\ref{tlb:datasets} summarizes the four datasets. As we can see, the four populations vary from one million in \texttt{TS} to 16 millions in \texttt{LA}, from the not-so-popular areas (e.g., \texttt{TS}) to popular areas (e.g., \texttt{LA}). Note that all the metropolitan population information is from the U.S. Census Bureau website.

\subsection{Conjecture Validation}

To validate the first conjecture above, we estimate the number of Twitter users for each area according to the eMarketer report claiming that 15.1\% of U.S. people are using Twitter as of Feb. 2014 \cite{emarketer14}. As we can see from Table~\ref{tlb:datasets}, the seed users range from 12.12\% in \texttt{LA} to 22.21\% in \texttt{CI} with the average ratio of 15.9\%. This result is consistent with the measurement in \cite{LiTow12} and implies that we have almost crawled all the users who have specified their city-level locations in these areas.

To validate the second conjecture above, we first define three locality metrics. In particular, for the multigraph $G=\langle V, E\rangle$ defined in Section~\ref{sec:DLA}, let $V'$ denote any subset of $V$. We define follower locality $l_\textrm{follower}(V')$, followee locality \newline $l_\textrm{followee}(V')$, and initiator locality $l_\textrm{initiator}(V')$ as
\begin{equation}
\begin{split}
l_\textrm{follower}(V') = \frac{|N^\mathcal{F}_I(V') \cap V'|}{|N^\mathcal{F}_I(V')|},
l_\textrm{followee}(V') = \frac{|N^\mathcal{F}_O(V') \cap V'|}{|N^\mathcal{F}_O(V')|},&\\
\textrm{and } l_\textrm{initiator}(V') = \frac{w(N^\mathcal{I}_O(V') \cap V')}{w\large(N^\mathcal{I}_O(V')\large)},&\label{eq:locality}
\end{split}
\end{equation}
respectively, where $N^\mathcal{F}_I(V'), N^\mathcal{F}_O(V')$, and $N^\mathcal{I}_O(V')$ represent the followers, followees, and initiators of $V'$, respectively, and $w(\cdot)$ represents the total weight of the corresponding interacting edges.


We let $V'$ equal the seed users in each area and then compute the corresponding locality. To do so, we crawl all the followers and followees of each seed user, and we also crawl the latest 600 tweets of each seed user to extract their initiators. For the comparison purpose, we build two types of random user sets. First, we merge the four seed sets into a single set from which we randomly select the same number of users as the seeds in each area. Second, we randomly select from the whole Twitter system the same number of users as the seeds in each area and compute their corresponding locality. We build 10 different user sets for both random user sets.

\begin{table}[t]
    \centering
    \caption{Locality in each area. Each element is composed of three values, representing the locality for the seed users in each area, the first type of random user set, and the second type of random user set, respectively.}
    \begin{tabular}{|c|c|c|c|}
       \hline
        $\mathbb{A}$ & $l_\textrm{follower}(U)$ (\%) & $l_\textrm{followee}(U)$ (\%) & $l_\textrm{initiator}(U)$ (\%) \\
       \hline
        \texttt{TS} &  8.1 | 0.08|0.04 & 9.2 | 0.3 | 0.1 & 12.8 | 0.5 | 0.2 \\
        \texttt{PI} &  4.9 | 0.4 | 0.2 & 8.4 | 1.5 | 0.6 & 14.9 | 2.7 | 1.2 \\
        \texttt{CI} &  6.9 | 1.0 | 0.5 & 10.3| 3.3 | 1.3 & 16.9 | 5.2 | 2.6 \\
        \texttt{LA} &  1.5 | 0.9 | 0.5 & 8.4 | 3.1 | 1.2 & 17.0 | 5.0 | 2.5 \\
        \hline
    \end{tabular}
    \label{tlb:locality}
\end{table}

\begin{table}[t]
    \centering
    \caption{Breaking down the initiator locality by three types of interactions. }
    \begin{tabular}{|c|c|c|c|}
       \hline
        $\mathbb{A}$ & Replying (\%) & Retweeting (\%) & Mentioning (\%)\\
       \hline
        \texttt{TS} &  14.41 & 10.61 & 13.37 \\
        \texttt{PI} &  14.05 & 12.95 & 16.99 \\
        \texttt{CI} &  17.46 & 14.23 & 18.50 \\
        \texttt{LA} &  15.43 & 15.44 & 19.05 \\
        \hline
    \end{tabular}
    \label{tlb:locInter}
\end{table}

Table~\ref{tlb:locality} shows the results of the locality analysis. We can see that the three locality values of the seed users in each area are always much higher than those of the random user sets. This result confirms our conjecture that physical proximity plays a big role in enabling online communications in Twitter. Moreover, Table~\ref{tlb:locality} shows a higher percentage of a user's initiators in the same area than that of his/her followees. It is not surprising because a user may follow many people in different areas but often interact with only a few selected followees. In addition, we can see that the followee locality is much higher than the follower locality except in \texttt{TS}. The reason can be explained as follows. A celebrity user such as @rihanna can easily attract millions of followers from around the world, but she may only follow relatively fewer people. So we can expect a higher percentage of her followees in the same area (Los Angeles) than that of her followers. Since each of the areas except Tuscon has a large number of celebrity users, the followee locality is much higher than the follower locality. In contrast, Tuscon is a much smaller area with relatively few celebrity users, so we can expect similar followee and follower locality. Table~\ref{tlb:locInter} also shows that mentions, replies, and retweets contribute similarly to the initiator locality of each seed set, so we do not distinguish them in the LocInfer design.

\begin{figure}[t]
\centering
    \includegraphics[width=0.4\textwidth]{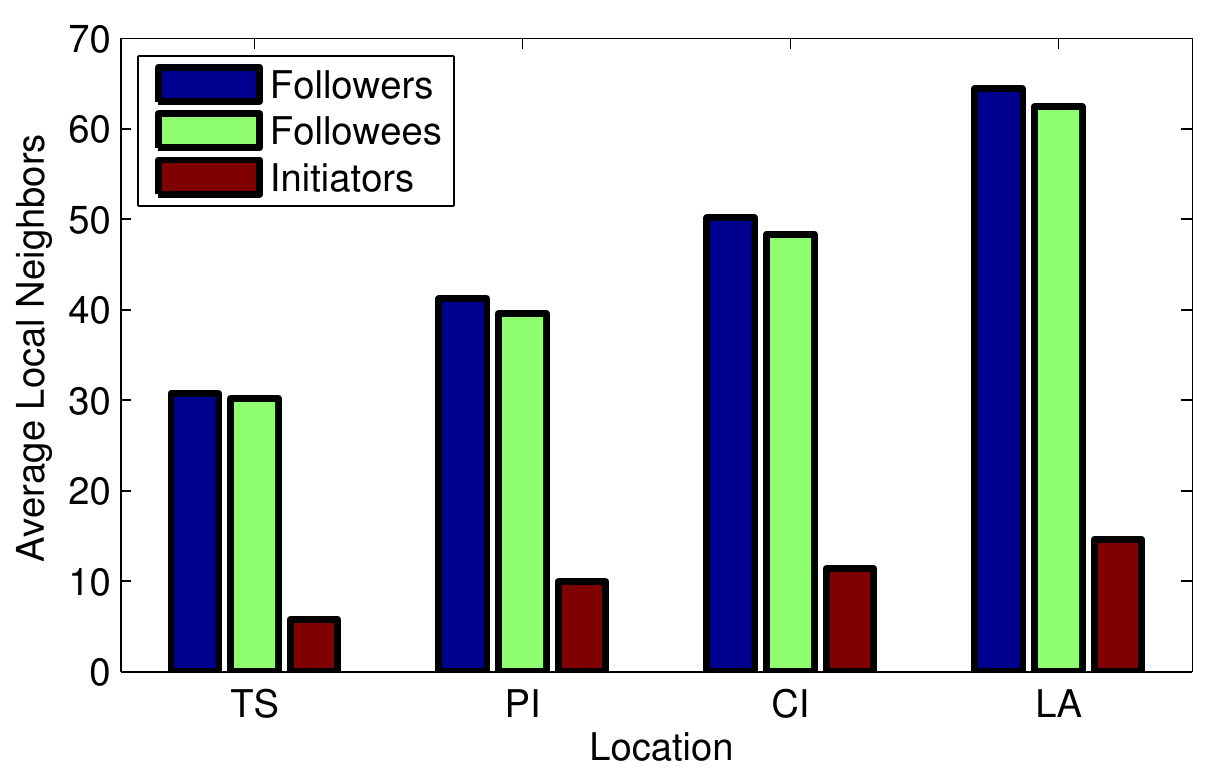}
\caption{The average local neighbors of the seed users.}\label{fig:avgd}
\end{figure}

Finally, although interacting communications (replies, mentions, and retweets) show much stronger locality than following communications, Fig.~\ref{fig:avgd} shows that the corresponding interacting edges (i.e., initiators) are much fewer than the following edges (i.e., followers and followees), meaning that people interact less than they follow others. Moreover, we also observe that people usually interact with the ones who they follow or follow them. In particular, let us define the overlap between $N^\mathcal{I}_O(V')$ and $N^\mathcal{F}_I(V') \cup N^\mathcal{F}_O(V')$ for each area as
\[
\frac{|N^\mathcal{I}_O(V')  \cap (N^\mathcal{F}_I(V') \cup N^\mathcal{F}_O(V'))|}{|N^\mathcal{I}_O(V')|}.
\]
Our analysis shows that the average overlap for the four areas is 96.2\%.

\section{{LocInfer}}\label{sec:LocInfer}
As stated before, our goal is to uncover the majority of Twitter users in an area $\mathbb{A}$. A naive solution is to use existing location inference methods \cite{HechtTwe11, ChengYou10, MahmuWhe12, LiTow12, MahmuHom14, BacksFin10, McGeeLoc13, JurgeTha13, YamagLan13, ComptGeo14} for estimating the location of every Twitter user and then select the ones in $\mathbb{A}$. However, these methods are impratical for our problem. In particular, they would require crawling the followers, the followees, and many tweets  for all the 255 million active Twitter users. Since Twitter has strict rate limits on data crawling \cite{api13}, the crawling process for these methods will be time-consuming. In addition, the network-based methods \cite{BacksFin10, McGeeLoc13, JurgeTha13, YamagLan13, ComptGeo14} need to store and process the edges of the whole Twitter graph, thus leading to prohibitive storage and processing costs.

Now we present LocInfer, an efficient and effective three-step system to identify the majority of users in $\mathbb{A}$. As mentioned earlier, LocInfer is built upon two conjectures which have been experimentally validated in Section~\ref{sec:Validation}. First, we can find a nontrivial number (15.9\% from our datasets) of users who have explicitly indicated a location in $\mathbb{A}$ through their personal profiles. These users are referred to as seed users (or seeds) in $\mathbb{A}$ and denoted by $S$. Second, user communications in Twitter exhibit strong geographic locality in the sense that users in the same area tend to have more intensive communications with each other in Twitter than with those from outside. Based on these two conjectures, LocInfer first builds a seed set $S$ (step 1 in Section ~\ref{sec:step1}) and then checks the one-hop neighbors of the seed set $S$, which constitute a candidate set denoted by $C$ (step 2 in Section ~\ref{sec:step2}). Because of nontrivial seed set $S$ and the strong geographic locality, $C$ will cover the majority of the users in $\mathbb{A}$, but also include many users outside. Hence LocInfer chooses the candidate users who have tight connections with $S$ as new seeds and add them to $S$, and this process continues until some termination conditions are met (step 3 in Section ~\ref{sec:step3}). The final seed set $S$ contains the majority of Twitter users in $\mathbb{A}$ with overwhelming probability. LocInfer is highly efficient because it only checks a much smaller set of Twitter users in contrast to all the Twitter users if existing methods \cite{HechtTwe11, ChengYou10, MahmuWhe12, LiTow12, MahmuHom14, BacksFin10, McGeeLoc13, JurgeTha13, YamagLan13, ComptGeo14} are applied.

We notice that many community structures (e.g., a group of people in different locations with common interests or past experience like classmates and colleagues) rather than the geographic community may also yield strong inter-connect- \newline ions. Hence LocInfer may include some users outside $\mathbb{A}$ in the candidate set $C$. However, the impact of such outside users is minimal because LocInfer only selects the users in the target area $\mathbb{A}$ as the seeds $S$ and only chooses the target users who have strong communications with $S$ later.


\subsection{Step~1: Finding Seed Users}\label{sec:step1}
The first step in LocInfer is to extract the seed users who are most certainly in $\mathbb{A}$. To that end, we use the same method as in Section~\ref{sec:Validation} by invoking the Twitter geo-search API to obtain the geo-profiled users and then refine them by checking their location descriptions to build the seed set $S$ in $\mathbb{A}$.

It is possible that some people may specify the fake home locations in their profiles, and it is infeasible to completely pinpoint and exclude such users. Fortunately, such self-reported locations have been verified to be very reliable \cite{ComptGeo14} and have been used as the ground truth in \cite{HechtTwe11, ChengYou10, MahmuWhe12, LiTow12, MahmuHom14, BacksFin10, McGeeLoc13, JurgeTha13, YamagLan13, ComptGeo14}. Meanwhile, we may accidentally exclude some users indeed in $\mathbb{A}$, which is quite acceptable given our focus on obtaining a reliable seed set in this step. We admit that more advanced methods can be used for the seed searching and refinement, which are left for the future work.

\subsection{Step 2: Finding Candidate Users} \label{sec:step2}

Based on the nontrivial number of seed users, the second step then is to construct a candidate-user set $C$ from the one-hop neighbors of $S$ that potentially covers the majority of Twitter users in $\mathbb{A}$ but is also much smaller than the set of all Twitter users. Below we first discuss how we decide the candidate users in $C$ and then theoretically analyze the coverage of $C$.

\subsubsection{Choosing $C$}

We first build the candidate set $C$ from the one-hop neighbors of $S$. The underlying intuition is based on the two conjectures validated in Section~\ref{sec:Validation}. Specifically, The second conjecture indicates that the users in the same geographic area tend to communicate more densely among themselves than to those from outside. On the one hand, if a user has very limited communications to all the seed users in $S$ which occupies about 15.9\% of the total users in $\mathbb{A}$, with high probability he/she is not in $\mathbb{A}$; on the other hand, a user that is indeed in $\mathbb{A}$ is very likely to have direct communication with some seeds. We therefore choose to build the candidate set $C$ from the one-hop neighbors of $S$, denoted as $N(S)$.

Two details need further consideration. As defined in Section~\ref{sec:DLA}, each Twitter user has four kinds of neighbors in $G=\langle V, E\rangle$: followers, followees, initiators, and responders. Which neighbors should we choose for each seed user? We observe from Fig.~\ref{fig:avgd} that many Twitter users may follow a large number of other users, but they tend to subsequently interact with relatively few followees. Since people usually interact with the ones who they follow or follow them (with averagely 96.2\% of overlap as stated in Section~\ref{sec:Validation}) and $C$ should cover as many users as possible in $\mathbb{A}$, we consider all the followers and followees of each seed user in this step. Moreover, since each user in Twitter can follow arbitrary users without prior consent, the unidirectional following relationship is not a reliable indicator of geographic closeness. To deal with this issue, we propose to only select the candidate users to be the followers and followees of each seed user in $S$ with each having at least $t$ followees and $t$ followers in $S$, where $t$ is a system threshold.


More formally speaking, for each user $u \in N^{\mathcal{F}}_O(S)\cup N^{\mathcal{F}}_I(S)$, we compute $n^{\mathcal{F}}_i(u) = |N^{\mathcal{F}}_I(u) \cap S|$ and $n^{\mathcal{F}}_o(u) = |N^{\mathcal{F}}_O(u) \cap S|$. If both $n^{\mathcal{F}}_i(u)$ and $n^{\mathcal{F}}_o(u)$ are no less than $t$, user $u$ is added to the candidate set $C$ and ignored otherwise.

Alg.~\ref{alg:step1} implements the overall process. Specifically, we first create a followee counter and a follower counter for each user in $N^{\mathcal{F}}_O(S)\cup N^{\mathcal{F}}_I(S)$. Then we traverse the followee and follower list of each seed and increase the corresponding followee and follower counters. If both the followee and follower counters exceed $t$, we choose the user $u$ as a candidate.

\begin{algorithm}
 \SetAlgoLined
 \SetKwInOut{Input}{input}\SetKwInOut{Output}{output}

 \Input{$S,N^{\mathcal{F}}_O(S), N^{\mathcal{F}}_I(S), t$}
 \Output{the candidate set $C$}

 $C \leftarrow \varnothing$;
 $c_o[u] \leftarrow 0, \forall u \in N^{\mathcal{F}}_O(S)$;
 $c_i[v] \leftarrow 0, \forall v \in N^{\mathcal{F}}_I(S)$\;
 \For{$u \in S$}{
    $c_o[v] ++, \forall v \in N^{\mathcal{F}}_O(u)$;
    $c_i[v] ++, \forall v \in N^{\mathcal{F}}_I(u)$\;
}


 \For{$u \in N^{\mathcal{F}}_O(S)$}{
    \If{$c_o[u]\geq t $ and $u \in N^{\mathcal{F}}_I(S)$ and $c_i[u]\geq t $}{
            $C \leftarrow C + \{u\}$\;
    }
 }
 \Return{$C$.}
  \caption{Obtain the candidate set $C$ by only checking the followee and follower lists of the seed set $S$.} \label{alg:step1}
\end{algorithm}

\subsubsection{Coverage of $C$}\label{sec:CoverageC}
The number of candidate users (i.e., $|C|$) is determined by both the number of seed users (i.e., $|S|$) and the system parameter $t$. A natural question is whether $C$ can cover the majority of users in target area $\mathbb{A}$. It is important because the new seeds (or equivalently the target users) will be found only from $C$.

To analyze the coverage of $C$, we first define the following terms and notation. We call users $i$ and $j$ \emph{mutual followers} if they follow each other. Let $G_{\mathbb{A}}=\langle V_{\mathbb{A}}, E_{\mathbb{A}}\rangle$ be a subgraph of the Twitter multigraph $G=\langle V, E\rangle$, where $V_{\mathbb{A}}\subseteq V$ is the set of the Twitter users in the target area $\mathbb{A}$, and $E_{\mathbb{A}}\subseteq E$ is the set of the directed following edges among the users in $V_{\mathbb{A}}$. Consider a seed set $S\subseteq V_{\mathbb{A}}$ with $s=|S|=\alpha |V_{\mathbb{A}}|$ users, where $\alpha\in (0,1]$. Let $N^t(S)$ denote the set of the followers and followees of $S$, each having at least $t$ followers and $t$ followees in $S$, where $t$ is the system threshold stated before. The coverage ratio of $C$ is defined as $r(t)=\frac{|N^t(S)\cup S|}{ |V_{\mathbb{A}}|}$. We then have the following theoretical results about the coverage of $C$ given $|S|$ and $t$.

\begin{theorem}
Assume that each user in $V_{\mathbb{A}}$ has on average $d_m$ mutual followers in $V_{\mathbb{A}}$. When $|V_{\mathbb{A}}|$ is large enough, the expected coverage ratio is $\overline{r}(t) \geq 1 - e^{- \alpha d_m}(1 - \alpha) \newline \sum_{i=0}^{t-1}\binom{s}{i}(\frac{p}{1-p})^i$, where $p=\frac{d_m}{|V_{\mathbb{A}}|-1}$ . \label{th:coverage}
\end{theorem}
\begin{proof}We first construct an undirected graph $G'=\langle V_{\mathbb{A}}, \newline E'\rangle$, where an edge $e'_{ij}\in E'$ is formed if and only if users $i$ and $j$ are mutual followers.
Let $N^{'t}(S)$ be the set of neighbors of $S$ in $G'$, each having at least $t$ neighbors in $S$. We proceed to define the coverage of $S$ in $G'$ as $r'(t) = |N^{'t}(S)\cup S|/|V_{\mathbb{A}}|$.

We now compute $r'(t)$. Since each user has on average $d_m$ edges in $E'$, the probability of one user connecting to any other user is $p=\frac{d_m}{|V_{\mathbb{A}}|-1}$. Moreover, since there are $s=\alpha |V_{\mathbb{A}}|$ seed users, the probability of any non-seed node $u$ connecting to less than $t$ seed users in $G'$ is given by
\begin{equation}\label{eq:rho}
\rho = \sum_{i=0}^{t-1}\binom{s}{i}p^i( 1- p)^{s-i} = ( 1- p)^{s} \sum_{i=0}^{t-1}\binom{s}{i}(\frac{p}{1-p})^i\;.
\end{equation}
When the number of users in $V_{\mathbb{A}}$ is large, we have
\[
\begin{split}
\lim\nolimits_{|V_{\mathbb{A}}| \to +\infty} ( 1- p)^{s}&=\lim\nolimits_{|V_{\mathbb{A}}| \to +\infty} ( 1- d_m/|V_{\mathbb{A}}|)^{\alpha |V_{\mathbb{A}}|} \\
&= e^{-\alpha d_m}\;.
\end{split}
\]
Since there are $|V_{\mathbb{A}}|-s$ non-seed users, the expected number of non-seed users connecting to $t$ or more seeds in $S$ can be computed as $(|V_{\mathbb{A}}|-s)(1-\rho)=|V_{\mathbb{A}}|(1-\alpha)(1-\rho)$. When $|V_{\mathbb{A}}|$ is large, we have
\[
\begin{split}
\overline{r'}(t)&=|N^{'t}(S)\cup S| / |V_{\mathbb{A}}|\\
&= 1-(1-\alpha)\rho\\
&\approx 1-(1-\alpha)e^{-\alpha d_m}\sum_{i=0}^{t-1}\binom{s}{i}(\frac{p}{1-p})^i.
\end{split}
\]

Since each edge in $E'$ corresponds to two directed edges in $E$, all the users in $N^{'t}(S)$ must belong to $N^{t}(S)$. On the other hand, a user in $N^{t}(S)$ may not appear in $N^{'t}(S)$. For example, consider a user who has exactly $t$ followers and $t$ followees in $S$ in graph $G$, where none of his followers and followees are the same. Then this user is an isolated vertex in $G'$, and he is certainly in $N^{t}(S)$ but not in $N^{'t}(S)$. Therefore, we have $N^{'t}(S)\subseteq N^t(S)$ and $\overline{r'}(t)\leq \overline{r}(t)$, and the theorem is proved.
\end{proof}

\begin{corollary}
$\overline{r}(t=1) \geq 1 - e^{- \alpha d_m}(1 - \alpha)$.\label{co:coverage1}
\end{corollary}

\begin{corollary}
$\overline{r}(t=2)\geq 1 - e^{- \alpha d_m}(1 - \alpha)(1+\alpha d_m)$. \label{co:coverage2}
\end{corollary}

Since $|V_{\mathbb{A}}|$ is often large in practice, Theorem~\ref{th:coverage} indicates that the coverage ratio $\overline{r}(t)$ approaches 1 when $\alpha d_m$ is large enough. Moreover, the choice of $t$ involves a tradeoff between the crawling cost and the coverage. Specifically, the larger the $t$, the fewer the candidates in $C$, the smaller the crawling cost, the more likely to miss some users in $\mathbb{A}$ (i.e., the lower coverage), and vice versa. The size of $S$ also affects the choice of $t$. On the one hand, if $S$ constitutes a relatively large portion of the users in $\mathbb{A}$ (say, $\alpha=30\%$), it may be safe to use larger $t$ because many users in $\mathbb{A}$ are more likely to have more followees and followers in $S$. On the other hand, if $S$ constitutes a relatively small portion of the users in $\mathbb{A}$ (say, $\alpha=10\%$), it may be safe to use smaller $t$ to avoid excluding too many users in $\mathbb{A}$. 

Here we illustrate how many seeds are needed to achieve a nearly 100\% coverage. Assume that each user in $\mathbb{A}$ has on average 15 mutual followers (i.e., $d_m = 15$). According to Corollaries~\ref{co:coverage1} and \ref{co:coverage2}, when $t=1$, 20\% of the users as seeds can cover 96.02\% of the target users in $\mathbb{A}$, and when $t=2$, 20\% and 30\% of the users as seeds can cover 84.07\% and 95.72\% of the users in $\mathbb{A}$, respectively. Similarly, if $d_m = 30$, only 10\% and 15\% of the users as seeds can cover 95.52\% and 95.99\% of the users for $t=1$ and $t=2$, respectively. These results indicate that when each user has sufficient mutual followers in $\mathbb{A}$, the followers and followees of a small number of seeds can cover the majority of the target users in $\mathbb{A}$.



\subsection{Step~3: Finding Target Users $U$} \label{sec:step3}

Although the candidate set $C$ covers nearly all the users in $\mathbb{A}$ for proper $t$, it may contain many users not in $\mathbb{A}$ who nevertheless have at least $t$ followees and also $t$ followers in the seed set $S$. For example, social butterflies \cite{YangAna12} or social capitalists \cite{GhoshUnd12} have been reported to automatically follow back whoever follows them, and users may also follow each other due to reciprocity \cite{YangAna12,GhoshUnd12}. We thus design the next step to identify the target user set $U$ in $\mathbb{A}$ from $C$ using both the following and interacting connections among the users.

Our key observation as stated is that each target user is very likely to demonstrate significant locality with the seed user set $S$. In other words, we expect that the target users form a strong local community with the seed users. From the initial seed set $S$, we iteratively check the candidate users in $C$, and the candidate who has the highest locality value with the seeds becomes a new seed and is added to $S$. The process iterates until certain conditions are met.

How should we compute the locality of Twitter users with diverse communications? Inspired by the Eq.~(\ref{eq:locality}), we consider three types of locality for any candidate user $u\in C$:
follower locality $l_\textrm{follower}(u)$, followee locality $l_\textrm{followee}(u)$, and initiator locality $l_\textrm{initiator}(u)$, which are computed as
\begin{equation}
\begin{split}
l_\textrm{follower}(u) = \frac{|N^\mathcal{F}_I(u) \cap S|}{|N^\mathcal{F}_I(u)|},
l_\textrm{followee}(u) = \frac{|N^\mathcal{F}_O(u) \cap S|}{|N^\mathcal{F}_O(u)|}, &\\
\textrm{and } l_\textrm{initiator}(u) = \frac{w\large(N^\mathcal{I}_O(u) \cap S)}{w(N^\mathcal{I}_O(u)\large)}, &\label{eq:interLoc}
\end{split}
\end{equation}
where $N^\mathcal{F}_I(u), N^\mathcal{F}_O(u)$, and $N^\mathcal{I}_O(u)$ are $u$'s followers, followees, and initiators, respectively, and $w(\cdot)$ denotes the total weight of the corresponding interacting edges. 

We also consider two methods to integrate the three types of locality. First, we choose the maximum one among them as $u$'s locality, i.e.,
\begin{equation}\label{eq:sepLoc}
l(u) = \max\{l_\textrm{follower}(u), l_\textrm{followee}(u), l_\textrm{initiator}(u)\}\;.
\end{equation}
Second, their weighted combination is used as the locality of $u$, i.e.,
\begin{equation}\label{eq:comLoc}
l(u) = \epsilon_1 l_\textrm{follower}(u) + \epsilon_2 l_\textrm{followee}(u) + \epsilon_3 l_\textrm{initiator}(u)\;,
\end{equation}
where $0\leq \epsilon_1,\epsilon_2,\epsilon_3\leq 1$ and $\epsilon_1 + \epsilon_2  + \epsilon_3 = 1$. In this paper, we choose each of them to be 1/3 for simplicity and leave other possible assignments as the future work.

Finally, we iteratively find the target users based on one of their five types of locality with regard to the seed set $S$. In each iteration, we compute the locality for each candidate $u\in C$ according to Eq.~(\ref{eq:interLoc}), Eq.~(\ref{eq:sepLoc}), or Eq.~(\ref{eq:comLoc}). The candidate with the highest locality is removed from $C$ and added to $S$ as a new seed, as this user contributes most to the tightness of the community around $S$. In addition, the follower, followee, and/or initiator locality values of the remaining candidates in $C$ need be updated in every iteration. Here we just use the followee locality to illustrate the updating operation. Let $l_\textrm{followee}^{(m)}(u)$ denote the followee locality for candidate $u$ in iteration $m\geq 0$, where $l^{(0)}(u)$ can be computed by using the initial seeds in $S$. Assuming that $u^\ast$ has been chosen as a new seed in iteration $m$, we update the followee locality for candidate $u$ as
\begin{equation}\label{eq:updateRatio}
l^{(m+1)}(u) = \left\{
    \begin{array}{ll}
        l^{(m)}(u) + 1/|N^{\mathcal{F}}_O(u)| & \text{if } u^* \in N^{\mathcal{F}}_O(u), \\
        l^{(m)}(u) & \text{o.w.}
    \end{array} \right.
\end{equation}
Follower and initiator locality can be updated similarly, and we may need to update the overall locality according to Eq.~(\ref{eq:sepLoc}) or Eq.~(\ref{eq:comLoc}). The iteration terminates when the seed set $S$ contains a desired number of users in $\mathbb{A}$, denoted by $\tau_A$. Then the sought target users correspond to all the users in $\mathbb{A}$. The complete process is summarized in Alg.~\ref{alg:step2}, which is implemented using a max-priority queue \cite{CLRS}.

\begin{algorithm}
 \LinesNumbered
 \SetKwInOut{Input}{input}\SetKwInOut{Output}{output}
 \Input{$S, C, \tau_\mathbb{A}$}
 \Output{$U$, i.e., the users in $\mathbb{A}$ }
 $U \gets S$\;
 Compute $l(u),\forall u\in C$, according to Eq. (\ref{eq:interLoc}), (\ref{eq:sepLoc}) or (\ref{eq:comLoc})\;
 $Q\gets \emptyset$\; 
 \For{$u \in C$}{
    $\mathsf{INSERT}(Q, u)$\;
    }
 \While{$|U| < \tau_\mathbb{A}$}{
    $u^* \gets \mathsf{EXTRAC}$\textendash$\mathsf{MAX}(Q)$\;
    $U \leftarrow U + \{u^*\},S \leftarrow S + \{u^*\}$\;
    \For{$u \in N^{\mathcal{F}}_I(u^*)$}{
       $\mathsf{INCREASE}$\textendash$\mathsf{KEY}(Q,u,l(u) + 1/|N^{\mathcal{F}}_O(u)|)$\;
    }
 }
 \Return{$U$.}
  \caption{Identify target users in $\mathbb{A}$ from $C$.} \label{alg:step2}
\end{algorithm}

The termination threshold $\tau_\mathbb{A}$ can be chosen in two ways. First, we can set $\tau_\mathbb{A}$ as the estimated number of Twitter users in $\mathbb{A}$, e.g., about 15.1\% of the population in $\mathbb{A}$ if $\mathbb{A}$ is in U.S. \cite{emarketer14}. Second, $\tau_A$ can be chosen according to the level of confidence we desire. In particular, our algorithm essentially ranks all the candidate users according to our confidence about their locations in $\mathbb{A}$. The later a candidate user is added to $U$, the lower confidence we have that he is indeed in $\mathbb{A}$. Therefore, if we want to obtain a set of target users in $\mathbb{A}$ with high confidence, a small $\tau_\mathbb{A}$ should be used; if we want to cover more users in $\mathbb{A}$, a larger $\tau_\mathbb{A}$ is suitable.

We now analyze the complexity of Alg.~\ref{alg:step2}. In Lines 4-5, we build a max-priority queue $Q$ based on each candidate's locality value, of which the complexity is $\mathcal{O}(|C|\log |C|)$. The loop beginning from Line 6 is used to find the target user one at a time. In each iteration, we extract the maximum value from the priority queue $Q$ in Line 7, set it as a new seed in Line 8, and update the locality value of all its followers in Lines 9-10. The complexity of Line 6-10 is $\mathcal{O}(\tau_\mathbb{A}d\log(|C|))$, where $d$ is the average degree in $\mathbb{A}$. Hence the overall complexity of Alg.~\ref{alg:step2} is $\mathcal{O}((|C|+\tau_\mathbb{A} d)\log(|C|))$.


One may wonder why we do not add more candidates to $C$ once a candidate is added as a new seed to $S$. We have shown in Section~\ref{sec:CoverageC} that the candidate users discovered through the initial seed set $S$ cover the majority of users in $\mathbb{A}$ with overwhelming probability. It is thus unlikely that we can identify more candidate users from newly identified seeds, which has been validated by our simulations in Section~\ref{sec:coverage}. We thus choose not to add more candidates in each iteration.


\subsection{Cost Analysis}\label{sec:CAR}
We now analyze the cost of LocInfer, which consists of the crawling cost and computation cost, and briefly compare it with the existing methods.

We first analyze the crawling cost of LocInfer, which is important given the tight rate limitations Twitter enforces on data crawling. First, Step~1 in LocInfer involves invoking the Twitter geo-search API continuously to obtain the initial seed set $S$ and needs to crawl some geo-tagged users' tweets. Second, Step~2 requires crawling the followees and followers of each seed user in $S$. Finally, Step~3 needs to crawl the followees, followers, and initiators of each candidate user in $C$. Recall that $d$ denotes the average number of followers and followees each seed user has. Our datasets in Section~\ref{sec:Validation} show that $d$ is approximately 600. It has also been reported that 15.1\% U.S. people use Twitter \cite{emarketer14} and that 15.9\% of Twitter users report city-level locations and become seeds in LocInfer. In LocInfer, a user is chosen as a candidate if he has $t$ followers and $t$ followees in $S$. So we can expect that the candidate set size $|C|$ is much smaller than $d|S|$, i.e., 14.4 times the population in the target area $\mathbb{A}$. In contrast, all previous (potential) solutions \cite{HechtTwe11, ChengYou10, MahmuWhe12, LiTow12, MahmuHom14, BacksFin10, McGeeLoc13, JurgeTha13, YamagLan13, ComptGeo14} involve crawling all the Twitter users. Thus LocInfer has a much smaller crawling cost, which makes it practical.

The computation cost of LocInfer is dominated by the third step with the complexity of Alg.~\ref{alg:step2} being $\mathcal{O}((|C| + \tau_\mathbb{A} d) \log(|C|))$, where $d$ is the average neighbors of each user and $\tau_\mathbb{A}$ is the number of target users in $\mathbb{A}$.

\subsection{Countermeasure}\label{sec:Counter}

LocInfer aims to discover the majority of users in any target area even if many of them do not disclose their locations explicitly in their personal profiles. We propose a simple countermeasure here to alleviate the possible concerns of some sensitive users about their location privacy. Since LocInfer discovers a user's location based on his tight connections with other users in the same area, the user can effectively hide his home location by following, retweeting, mentioning, and replying Twitter users outside his home area on a regular basis. This strategy is meaningful because people can follow or interact with others who are in different areas but share the same interests. For example, a user in New York City and the other in Los Angeles may interact in Twitter because they were university classmates in Dallas or knew each other in a concert. The efficacy of this countermeasure is evaluated in Section~\ref{sec:cm}.

\section{Performance Evaluation}\label{sec:Evaluate}

In this section, we thoroughly evaluate LocInfer. As stated before, this paper targets a different problem with existing work \cite{HechtTwe11, ChengYou10, MahmuWhe12, LiTow12, MahmuHom14, BacksFin10, McGeeLoc13, JurgeTha13, YamagLan13, ComptGeo14}, and hence we will not compare LocInfer with them head to head but could incorporate with them in our future work. 

\subsection{Methodology} \label{sec:methodology}
To evaluate LocInfer, we first need build a testing multigraph $G=\langle V, E\rangle$ formed by both users known to be and not be in a target area $\mathbb{A}$, where one challenge is that we cannot directly determine all the Twitter users in $\mathbb{A}$.

To tackle this challenge, we adopt the method used by existing work \cite{HechtTwe11, ChengYou10, MahmuWhe12, LiTow12, MahmuHom14, BacksFin10, McGeeLoc13, JurgeTha13, YamagLan13, ComptGeo14}. Specifically, since the self-reported locations have been found reliable \cite{ComptGeo14}, for each area $\mathbb{A}$ in Table~\ref{tlb:datasets}, we treat all the seed users in $S$ discovered in the first step as the positive ground truth (i.e., they are indeed in $\mathbb{A}$) and randomly partition $S$  into a seed subset $\overline{S}$ of size $\alpha|S|$ and a testing subset $T$ of size $(1-\alpha)|S|$.

For the negative ground truth, we check the followers and followees of $S$ and record the set of users who have specified a location outside $\mathbb{A}$ and randomly choose $\beta$ fraction of these users, where $\beta$ is set as the ratio of seed users over the estimated number of Twitter users in $\mathbb{A}$, as shown in the fourth column of Table~\ref{tlb:datasets}. We denote by  $\Theta$ the resulting user set and let $V=S\cup \Theta$. We finally compute edges among all the users in $V$ according to their followings and interactions by analyzing their followers, followees, and the latest 600 tweets.

We then apply LocInfer to the testing multigraph $G$. Specifically, we first use $\overline{S}$ as the seed set and apply Alg.~\ref{alg:step1} to generate the candidate set $C$. We then apply Alg.~\ref{alg:step2} to  $C$ to generate the target user set $U$ by choosing a $\tau_{\mathbb{A}}$. Following the definitions in Section~\ref{sec:DLA}, the \emph{coverage} can be computed as $|U\cap S|/|S|$, and the \emph{accuracy} can be computed as $|U \cap S|/|U|$ ($|U| = \tau_{\mathbb{A}}$).

Unless stated otherwise, we choose $t=2$ when building the candidate set $C$ with Alg.~\ref{alg:step1} for \texttt{LA} and $t=1$ for all other three datasets, and set $\alpha=0.159$, the average ratio for the four datasets in Table~\ref{tlb:datasets}. The testing multigraphs are summarized in Table~\ref{tlb:evaluation}.

\begin{table}[!th]
    \centering
    \caption{The testing multigraphs for the evaluation. ($\alpha = 0.159$) }
    \begin{tabular}{|c|c|c|c|c|c|}
       \hline
        $\mathbb{A}$ & $|S|$ & $|\overline{S}|$& $|T|$ & $\beta$ & $|\Theta|$\\
       \hline
        \texttt{TS} &  28,161 & 4,478  & 23,683 & 18.65\% & 162,446 \\
        \texttt{PI} &  144,033 & 22,901 & 121,132 & 15.9\% &630,321 \\
        \texttt{CI} &  318,632 & 50,662 & 267,970 & 22.21\% & 1,529,431 \\
        \texttt{LA} &  300,148 & 47,724 & 252,424 & 12.12\% &710,085 \\
        \hline
    \end{tabular}
    \label{tlb:evaluation}
\end{table}

\begin{figure}[b]
\centering
    \includegraphics[width=0.4\textwidth]{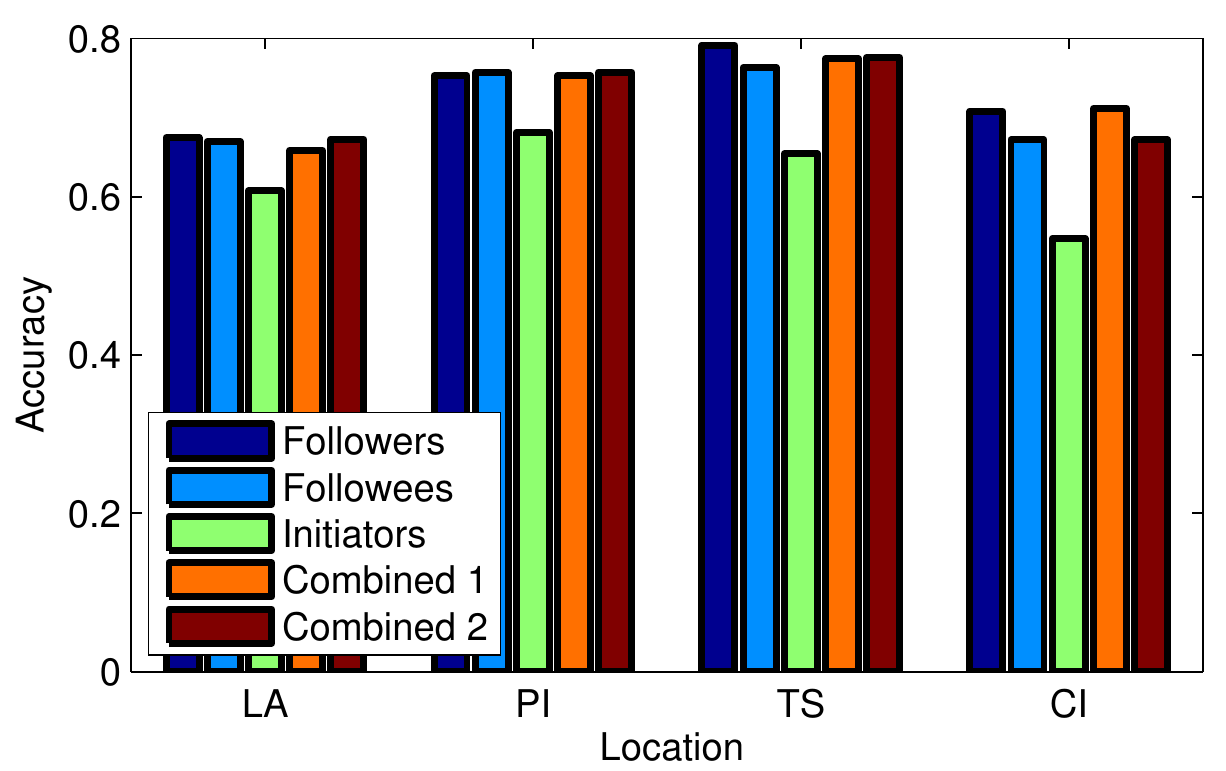}
\caption{The accuracy of LocInfer.}\label{fig:acc}
\vspace{-0.2in}
\end{figure}

\begin{figure*}[t]
\centering
    \subfigure[\texttt{LA}]{\label{fig:acc-la}
        \includegraphics[width=0.23\textwidth]{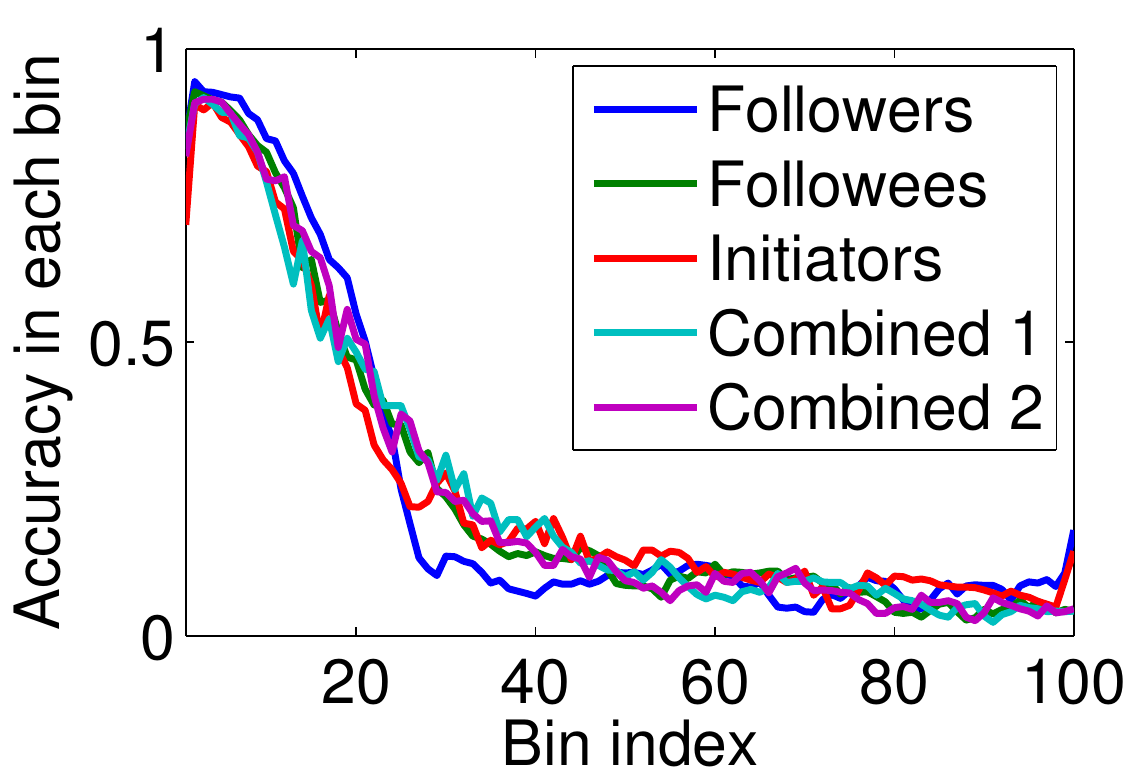}} \hfill
    \subfigure[\texttt{PI}]{\label{fig:acc-pi}
        \includegraphics[width=0.23\textwidth]{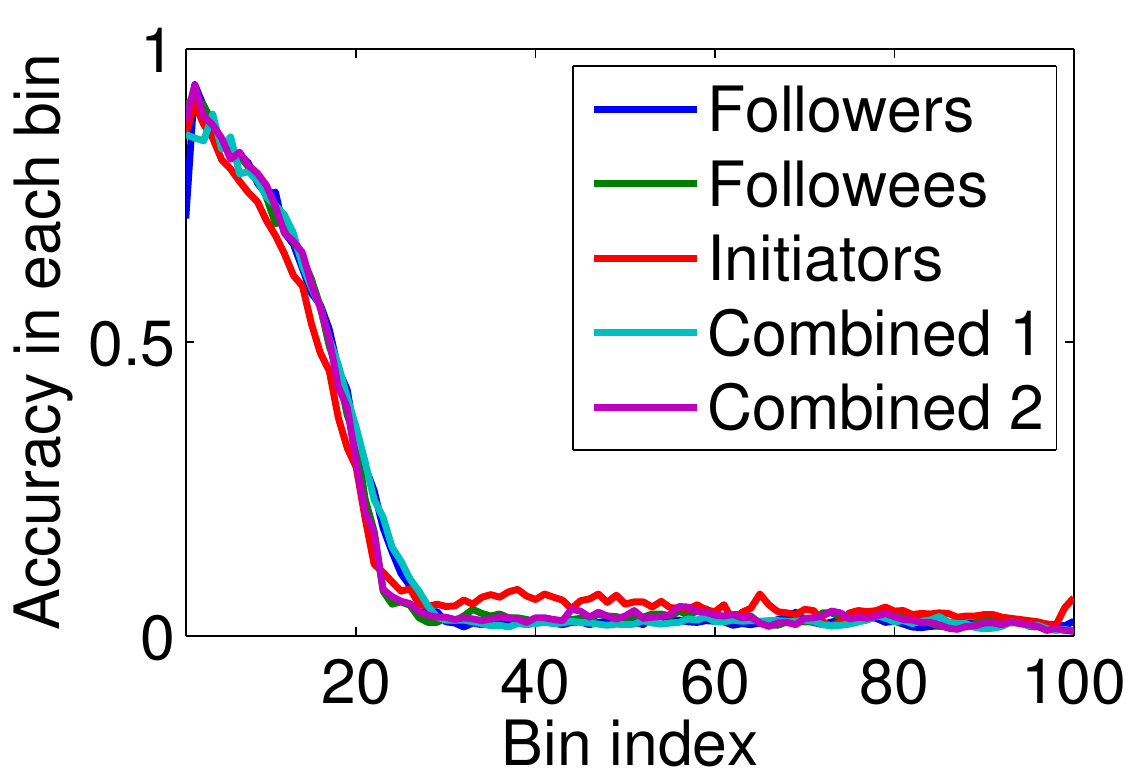}} \hfill
    \subfigure[\texttt{TS}]{\label{fig:acc-ts}
        \includegraphics[width=0.23\textwidth]{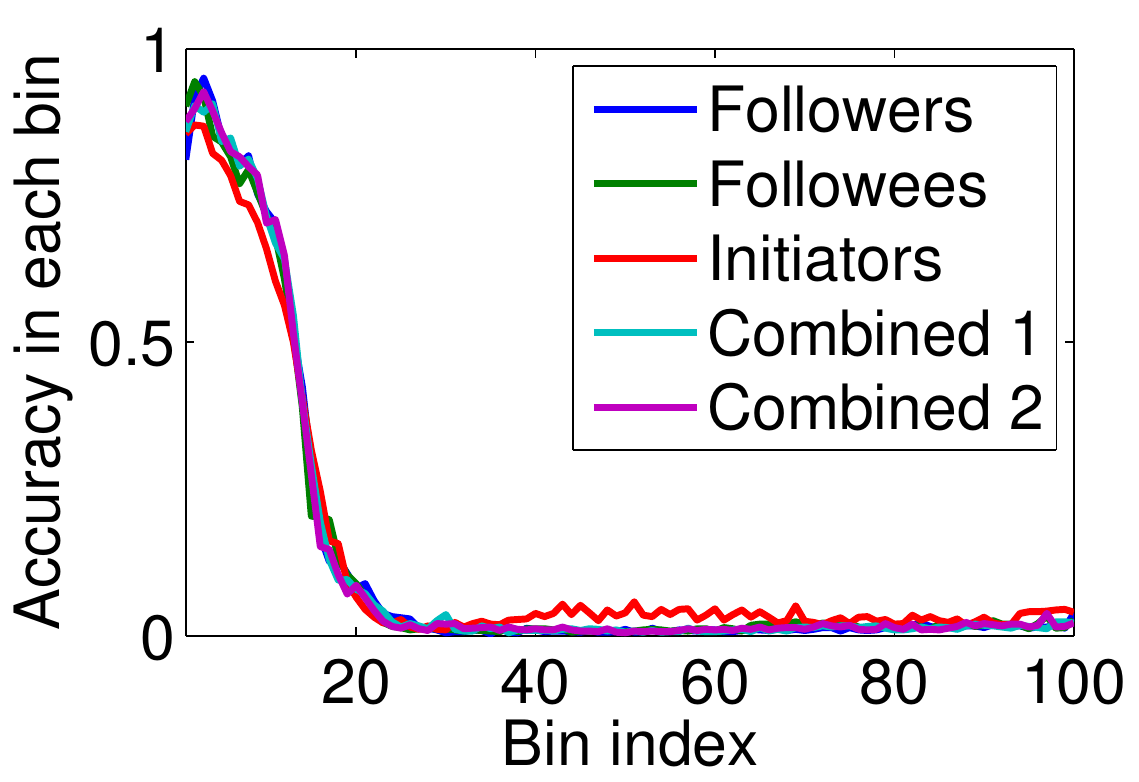}} \hfill
    \subfigure[\texttt{CI}]{\label{fig:acc-la}
        \includegraphics[width=0.23\textwidth]{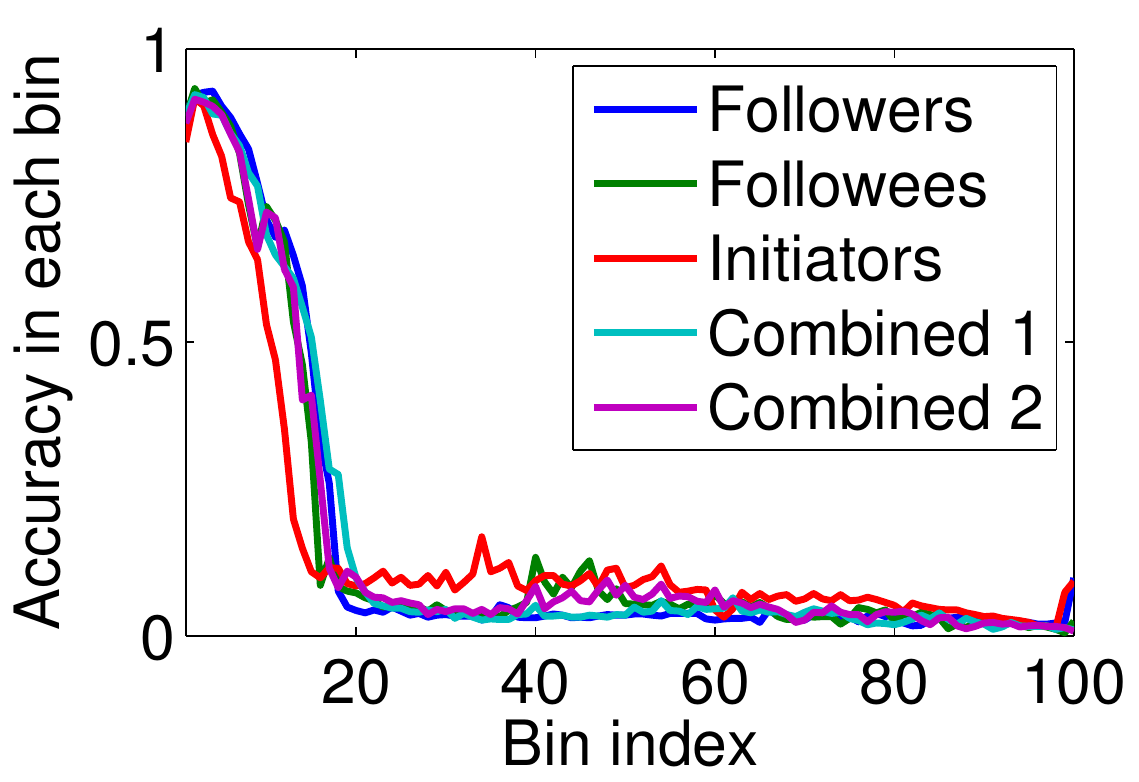}} \hfill
    \caption{Detailed accuracy illustration.}
    \label{fig:accu-bin}
\vspace{-0.1in}
\end{figure*}

\subsection{Accuracy} \label{sec:accuracy}

We first evaluate the accuracy of LocInfer. We compute five locality values for each user, including follower locality, followee locality, initiator locality, and the two locality values defined in Eq.~(\ref{eq:sepLoc}) and Eq.~(\ref{eq:comLoc}), respectively.


Fig.~\ref{fig:acc} shows the accuracy of LocInfer for the four datasets, where $\alpha=|\overline{S}|/|S|=0.159$ and $\tau_\mathbb{A}=|S|$. We can see that the five locality metrics all lead to high accuracy in each area, and initiator locality has the worst performance among them. Specifically, the average accuracy of four datasets for each locality are 73.2\%, 72.6\%, 62.3\%, 72.4\%, 71.9\%, respectively. The reason is that initiator locality depends on interacting edges (corresponding to replies, mentions, and retweets) which are much sparser than following edges in the directed Twitter multigraph as shown in Fig.~\ref{fig:avgd}. Therefore, if many users in $\mathbb{A}$ only follow many people but do not interact with them subsequently, they may be reachable from seed users through following edges but not from interacting edges. We will show the coverage for different locality metrics in the following Section~\ref{sec:coverage}. Moreover, the locality defined in Eq.~(\ref{eq:sepLoc}) and Eq.~(\ref{eq:comLoc}) have nearly the same accuracy with both the follower and followee locality. This is expected because about 96.2\% of the seed set's initiator neighbors are from their followers or followees, as indicated in Section~\ref{sec:Validation}. 

To shed more light on the accuracy of LocInfer, we set $\tau_A=|C|$ so that $U=C\cup \overline{S}$ when Alg. 2 terminates, i.e., every candidate user is eventually added into $S$. Let $U'$ denote the newly discovered users (may not in $\mathbb{A}$), i.e., $U'=C$. We partition $U'$ into 100 bins of equal size $|U'| / 100$ according to the order they are added, where the bins of smaller indexes contain the users discovered earlier. Let $x_i$ denote the number of positive ground-truth users in the $i$-th bin. Fig.~\ref{fig:accu-bin} shows the accuracy of the $i$-th bin, which is defined as the ratio of the number of positive ground-truth users in the $i$-th bin and the number of users in each bin and is computed as $100x_i/|U'|$. We can see that the accuracy in each bin decreases as the bin index increases, which is expected, as the later the users are added to $U'$, the less likely they are indeed located in $\mathbb{A}$.


\begin{figure}[t]
\subfigure[\texttt{TS}]{ \label{fig:impact-alpha-ts}
\includegraphics[width=0.22\textwidth]{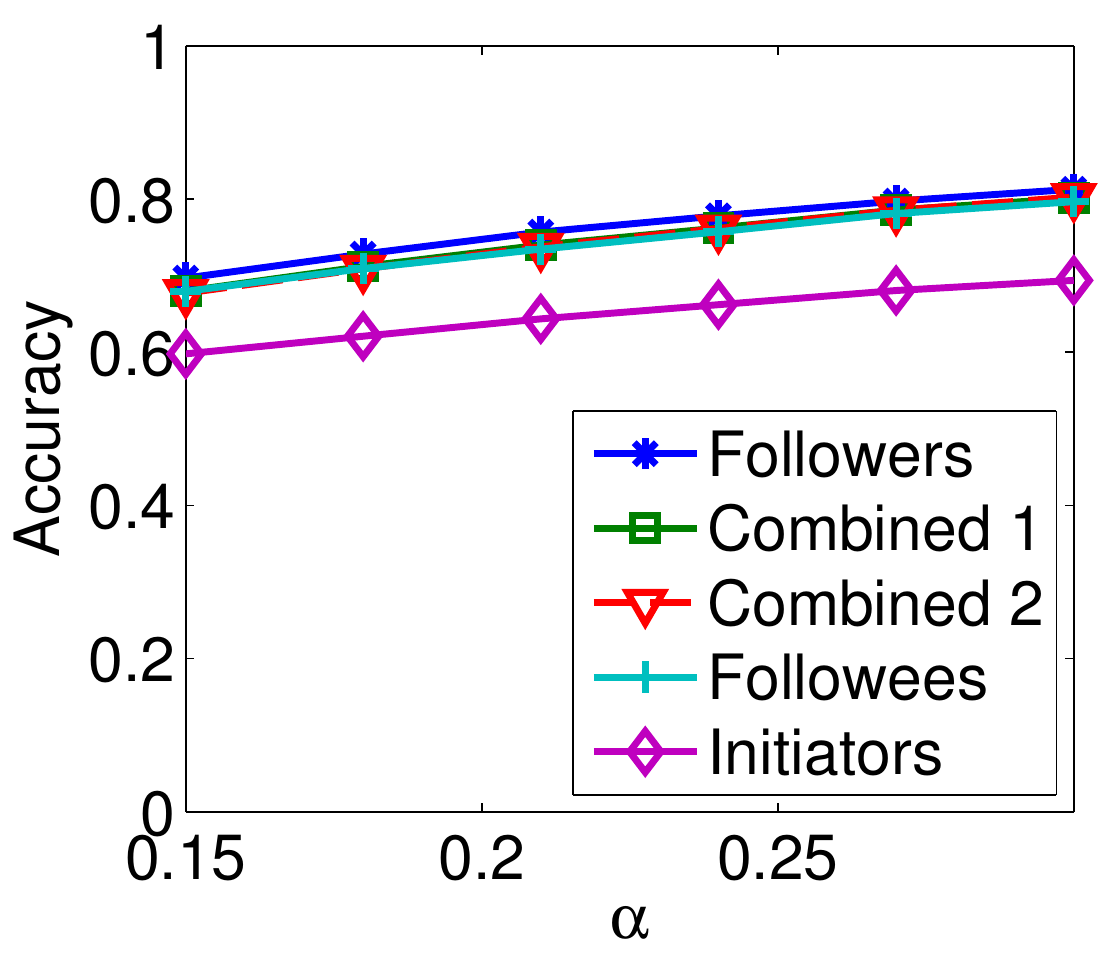}} \hfill
\subfigure[\texttt{PI}]{ \label{fig:impact-alpha-pi}
\includegraphics[width=0.22\textwidth]{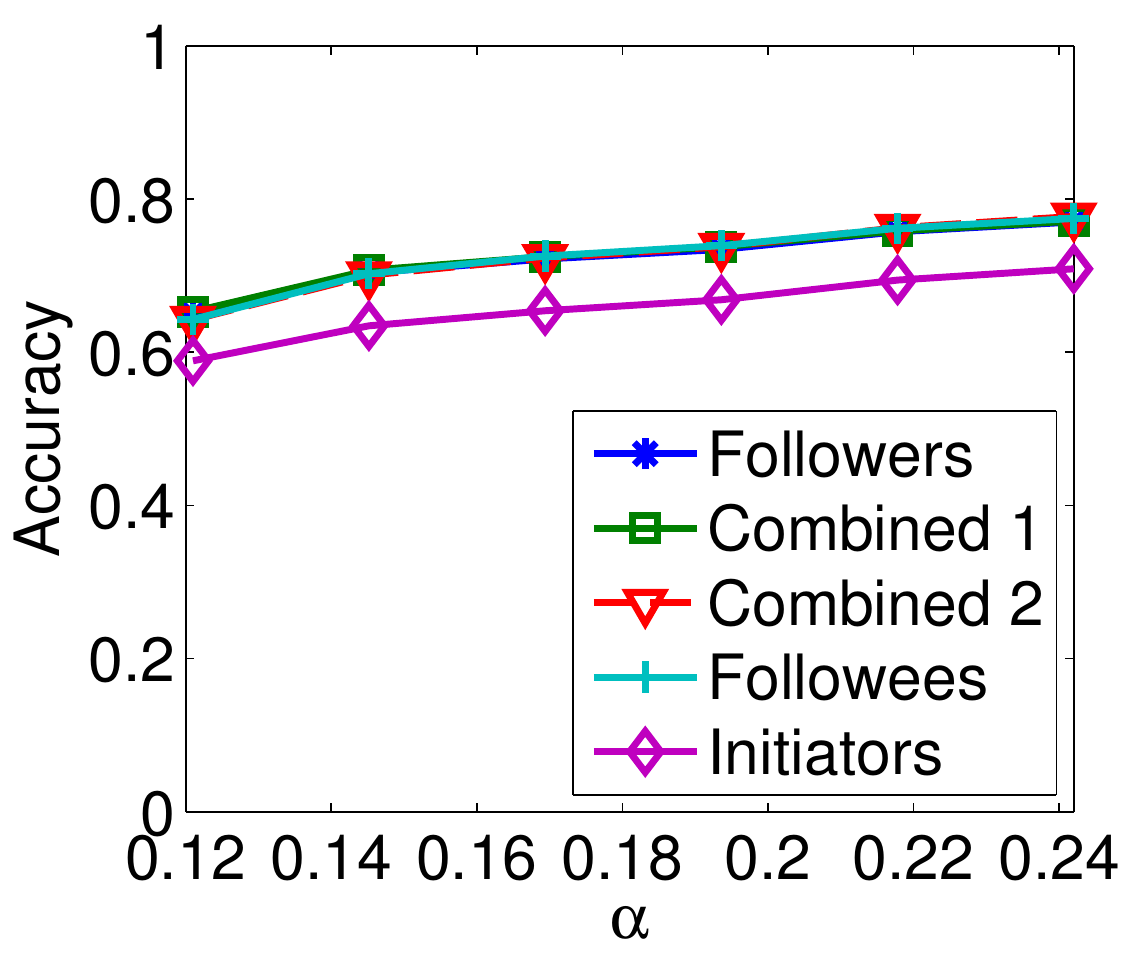}} \hfill
\caption{The impact of $\alpha$.}
\label{fig:impact-alpha}
\vspace{-0.2in}
\end{figure}

Fig.~\ref{fig:impact-alpha} shows the impact of $\alpha = |\overline{S}|/|S|$ on the accuracy of LocInfer. As expected, the accuracy under all locality metrics increases as $\alpha$ increases. The reason is that the larger the $\alpha$, the more seeds, and the easier the target users in $\mathbb{A}$ can be discovered. The downside is that more seeds lead to a larger candidate set and thus higher crawling and computational cost, as Alg.~\ref{alg:step1} needs to check all the neighbors of the seeds.


\begin{figure}[t]
\centering
    \subfigure[For different locations.]{\label{fig:t-loc}
        \includegraphics[width=0.22\textwidth]{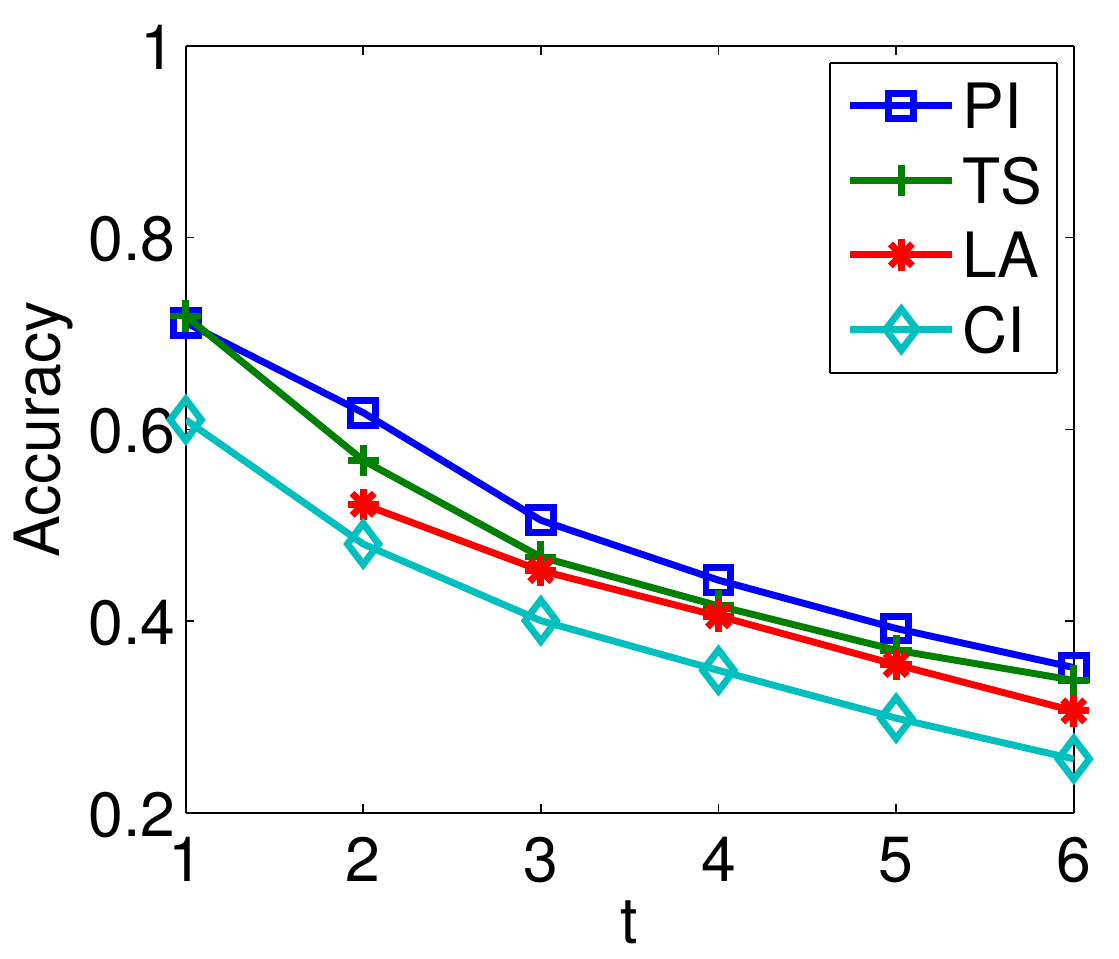}} \hfill
    \subfigure[For different localities.]{\label{fig:t-type}
        \includegraphics[width=0.22\textwidth]{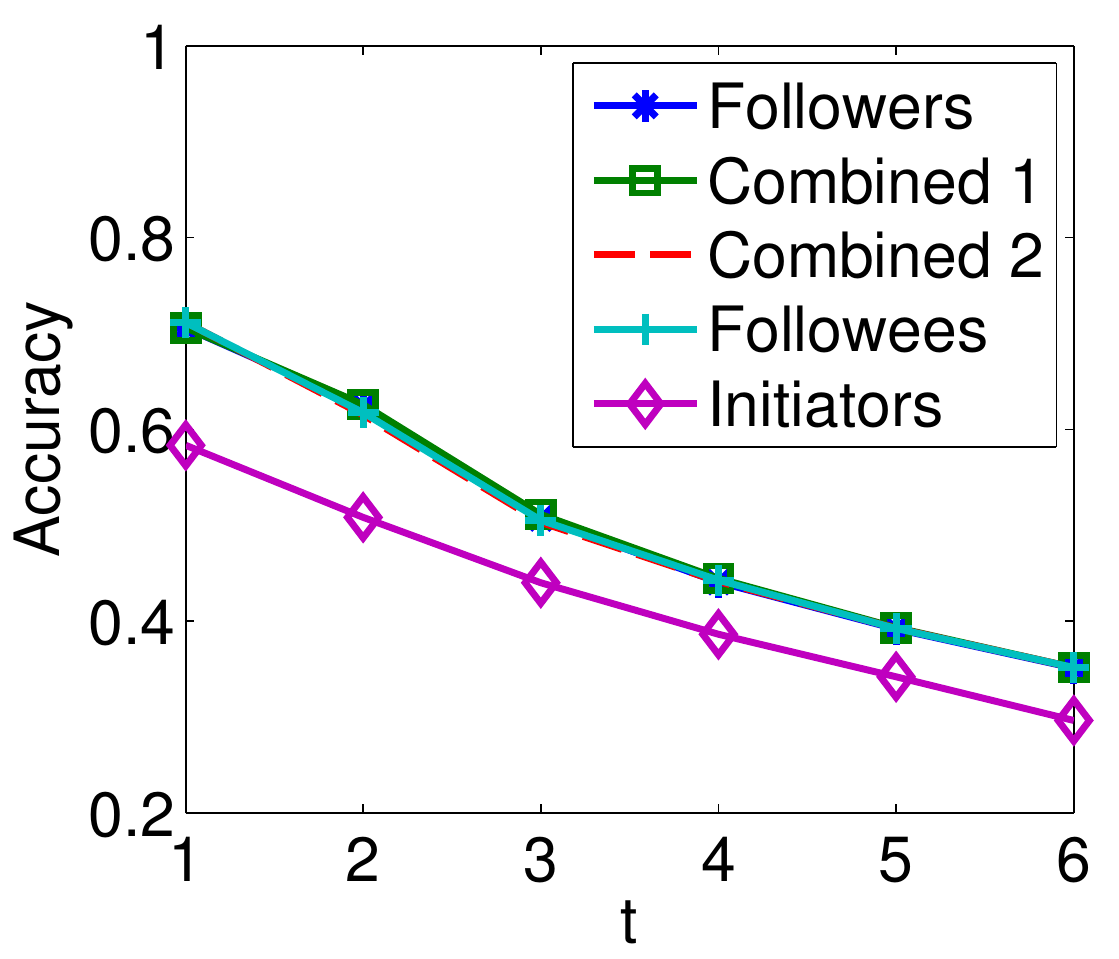}} \hfill
\caption{The impact of $t$.}
\label{fig:impact-t}
\vspace{-0.2in}
\end{figure}

Fig.~\ref{fig:impact-t} shows the impact of $t$ on the accuracy by varying $t$ from one to six. Specifically, Fig.~\ref{fig:t-loc} shows the accuracy for four areas using the followee locality, while Fig.~\ref{fig:t-type} shows the accuracy for different locality metrics by using the \texttt{PI} dataset. Both figures show the accuracy decreases as $t$ increases. This is expected because increasing $t$ will result in the decrease in the size of candidate set and hence miss more users in the target user list who have no chance to appear in the candidate set. However, there is a tradeoff between the accuracy and cost because smaller candidate set will also bring the lower crawling and computational cost.


\begin{figure}[t]
\subfigure[Followees]{\label{fig:tradeoff-friends}
    \includegraphics[width=0.23\textwidth]{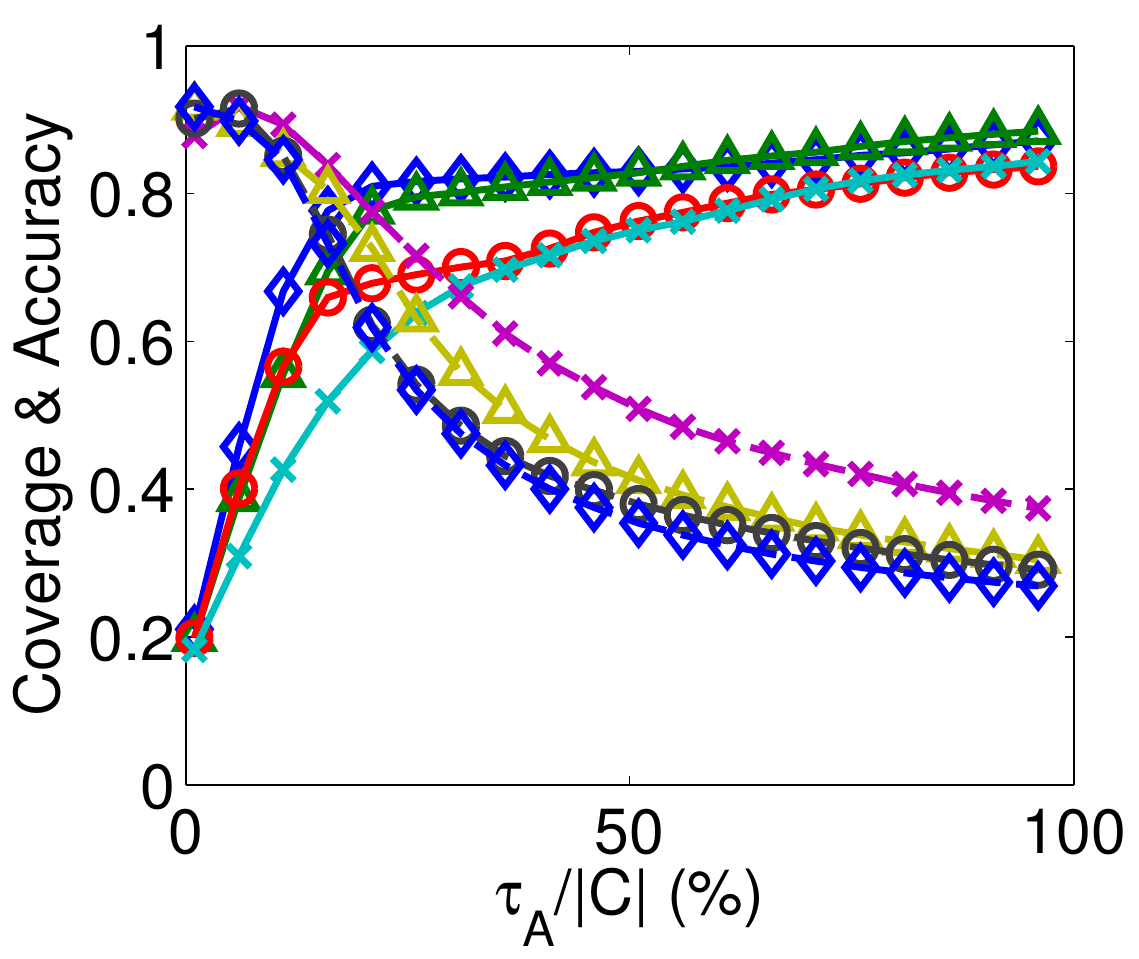}}
\subfigure[Initiators]{\label{fig:tradeoff-initiators}
    \includegraphics[width=0.23\textwidth]{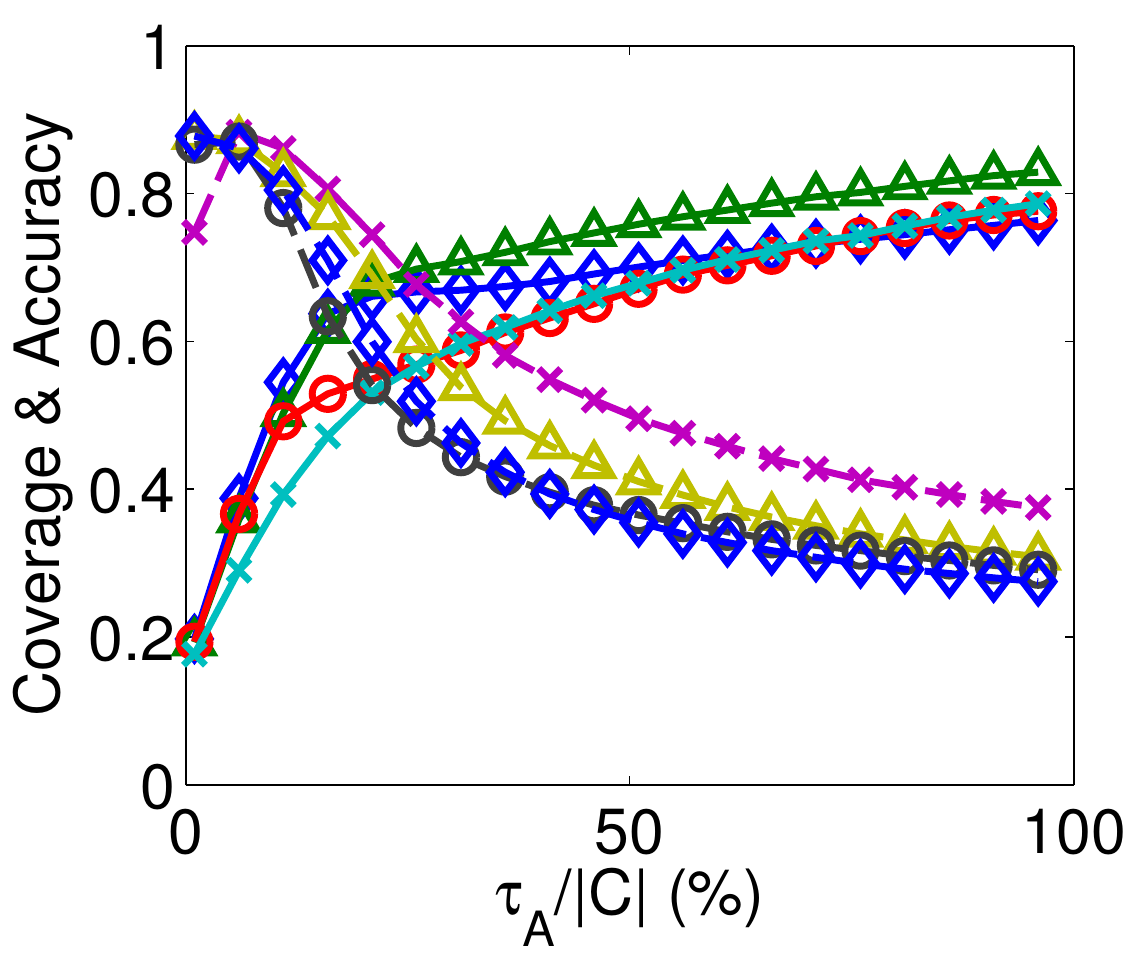}}
\caption{The tradeoff between the coverage and accuracy. The solid and dash curves are the coverage and accuracy; the marks $\diamond, \triangle, \circ, \times$ represent \texttt{TS}, \texttt{PI}, \texttt{CI}, and \texttt{LA}, respectively.}
\label{fig:tradeoff}
\vspace{-0.2in}
\end{figure}

\subsection{Coverage} \label{sec:coverage}

Fig.~\ref{fig:tradeoff} shows the coverage of LocInfer when $\alpha = 0.159$ with the desired number of target users (i.e., $\tau_\mathbb{A}=|U|$), varying from zero to the whole candidate set size $|C|$. We use both the followee and follower locality in this experiment. As expected, the larger $\tau_\mathbb{A}$, the more users in $T$ contained in $U$, the higher coverage, and vice versa. When we set $\tau_\mathbb{A} = |C|$, the average coverage of these four locations by using followee, follower, and initiator locality is equal to 86.3\%, 86.6\%, and 79.7\%, respectively. As stated, since the interacting edges (corresponding to replies, mentions, and retweets) are much sparser than following edges in the directed Twitter multigraph as shown in Fig.~\ref{fig:avgd}, the initiator locality has less coverage than the followee and follower locality. Moreover, the average coverage by using the follower or followee locality in Fig.~\ref{fig:tradeoff} is consistent with Corollary~\ref{co:coverage1}. Specifically, the average number of mutual followers $d_m$ for four datasets is 7.8, 9.0, 11.6, and 11.6, respectively. According to Corollary~\ref{co:coverage1}, when $\alpha = 0.159$, $\overline{r}(t=1) \geq 82.3\%$ which coincides with our results.

\subsection{Accuracy and Coverage Tradeoff}
Fig.~\ref{fig:tradeoff} also shows the anticipated tradeoff between the coverage and accuracy. As we can see, the larger $\tau_{\mathbb{A}}$, the more the positive ground-truth users will be added to $U$, resulting in higher coverage. However, a larger $\tau_{\mathbb{A}}$ will also introduce negative ground-truth users into $U$, resulting in lower accuracy. This tradeoff could guide us to choose the parameter $\tau_{\mathbb{A}}$. On the one hand, if one desires higher coverage, a large termination threshold $\tau_A$ should be used, but it is possible that many users in $U$ may be not indeed in $\mathbb{A}$. On the other hand, if one wants to be certain that the users discovered by LocInfer are most likely in $\mathbb{A}$, a smaller $\tau_A$ should be used at the cost of possibly missing some users indeed in $\mathbb{A}$.


\begin{figure}[b]
\centering
    \subfigure[Followees]{\label{fig:cm-friends}
        \includegraphics[width=0.23\textwidth]{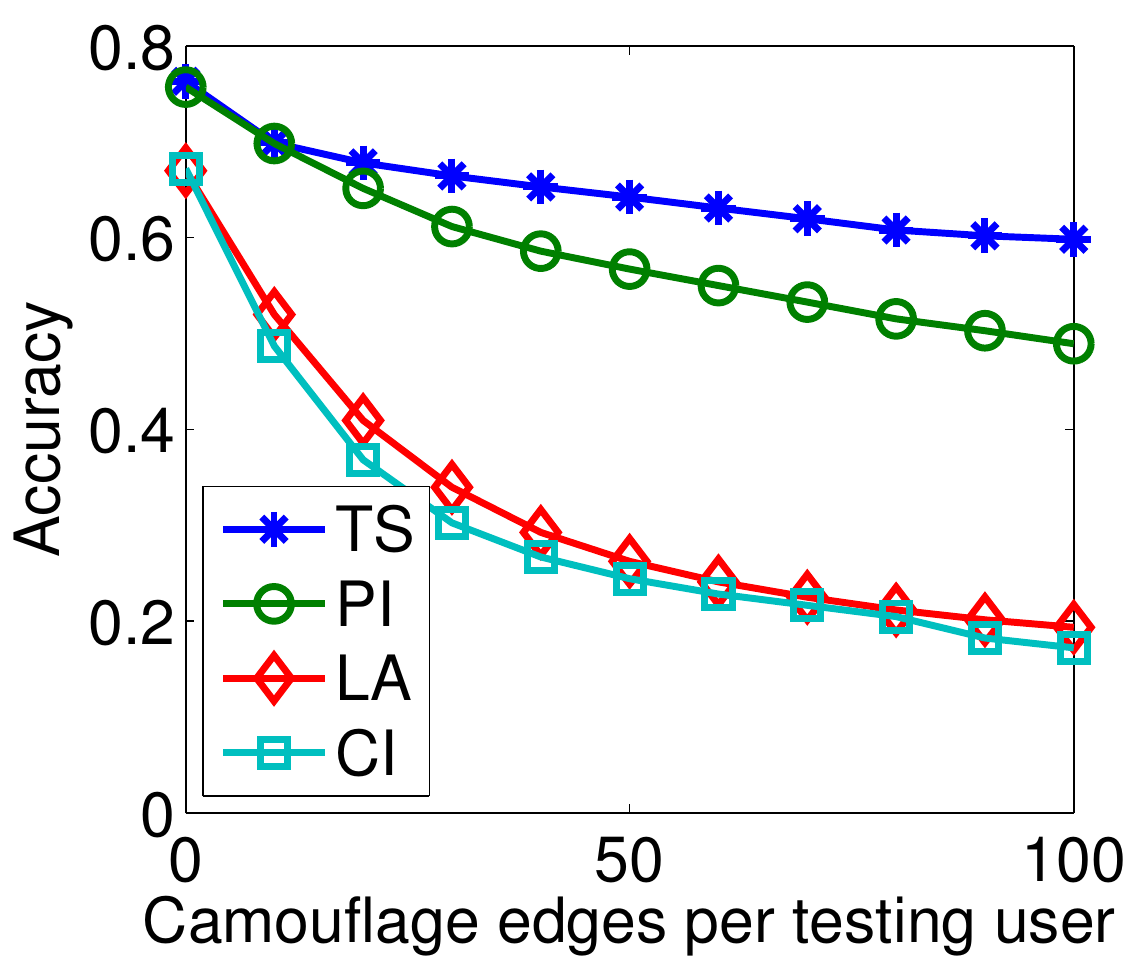}} \hfill
    \subfigure[Followers]{\label{fig:cm-followers}
        \includegraphics[width=0.23\textwidth]{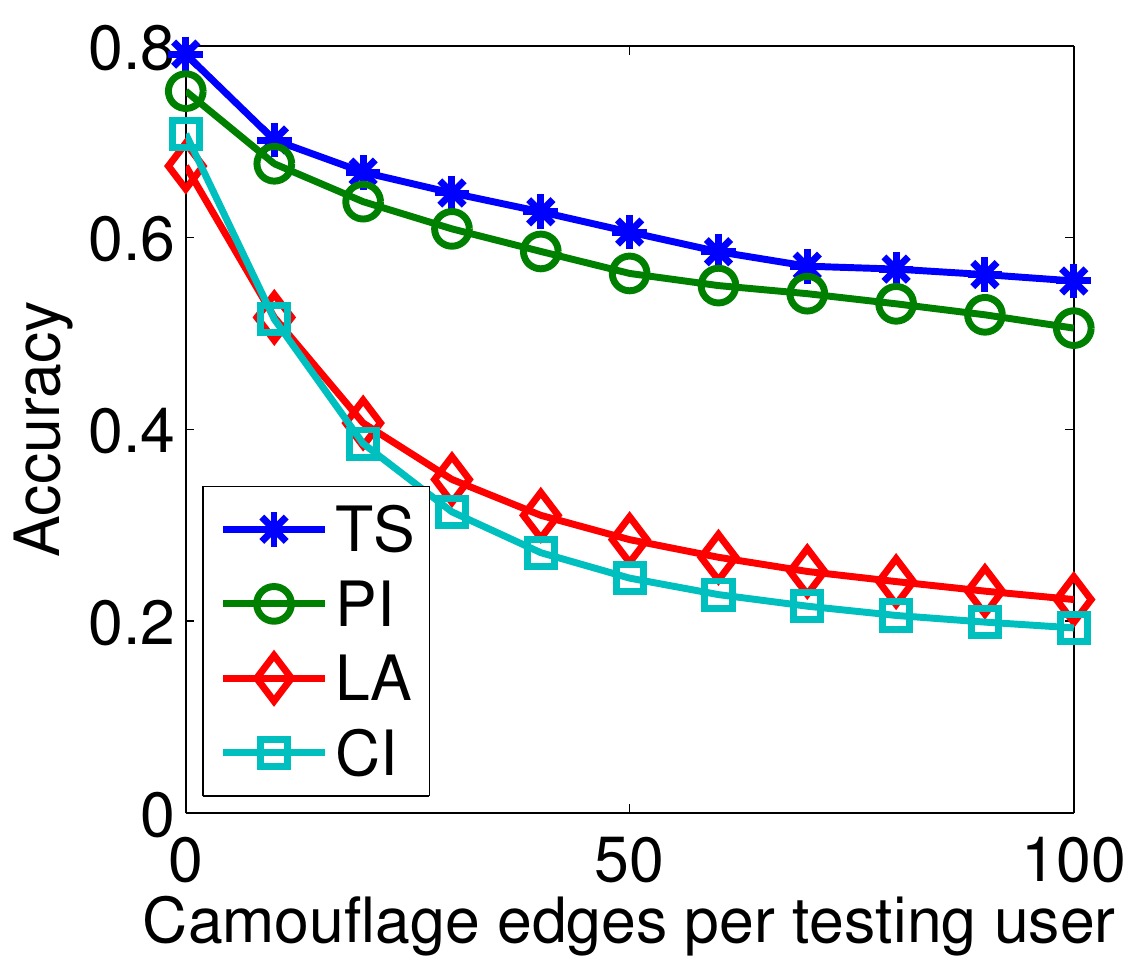}} \hfill
\caption{Countermeasure efficacy. }
\vspace{-0.1in}
\label{fig:cm}
\end{figure}

\subsection{Effectiveness of Countermeasure} \label{sec:cm}

To evaluate the efficacy of this countermeasure, we let each user in the testing set $T$ in each area additionally follow or be followed by a certain number of users from $\Theta$ who are not in $\mathbb{A}$, and we refer to those following edges as \emph{camouflage} edges. Fig.~\ref{fig:cm} shows the accuracy result under this countermeasure. As we can see, the accuracy of LocInfer decreases as the number of camouflage edges increases, highlighting the efficacy of the countermeasure. Besides adding random following edges, a user can also retweet, mention, and reply to random users on a regular basis to counteract LocInfer, which is expected to yield the similar results as these interactions can also decrease the geographic locality.

\section{Related Work}\label{sec:RW}
In this section, we briefly present the existing work mostly related to this paper.

Inferring a Twitter user's hidden location has been widely studied in the community, which can be categorized as cont- \newline ented-based and network-based methods. Content-based methods \cite{HechtTwe11, ChengYou10, MahmuWhe12, MahmuHom14} try to infer the user's location by his tweets. For example, Cheng \emph{et al.} \cite{ChengYou10} proposed a probabilistic framework to estimate a Twitter user's location based on his tweets, resulting in placing 51\% of Twitter users within 100 miles of their home locations. Mahmud \emph{et al.} \cite{MahmuHom14} further improved this result to 64\% for city-level location inference. Hecht \emph{et al.} \cite{HechtTwe11} thoroughly studied the location profiles for the Twitter users and found that 34\% of the users either left them empty or just non-geographic information. They also inferred the user's country and state information by checking their tweets. Network-based methods try to estimate a Twitter user's locations by his neighbors \cite{BacksFin10, McGeeLoc13, JurgeTha13, YamagLan13, ComptGeo14}. Jurgens \cite{JurgeTha13} aimed to infer all the users' location by building a global networks and then propagating location assignments from several seeds. Yamaguchi \emph{et al.} \cite{YamagLan13} built several distributed landmarks and then inferred a user's location based on the connections with them. Compton \emph{et al.} \cite{ComptGeo14} inferred the locations of all the users in Twitter by minimizing their distances with the labelled users. Moreover, Li \emph{et al.} \cite{LiTow12} combined the content and network information to obtain the more accurate estimation. All these schemes seek to address the same question: how can we infer a user's hidden location from all his location-related tweets and/or neighbors' locations? This paper targets a different problem: could we discover all or the majority of Twitter users in a metropolitan area? Directly adopting these existing methods to address our problem will result in scanning the whole Twitter network. Moreover, the accuracy of LocInfer outperforms the state of the art in \cite{MahmuHom14}.


This paper is also related to privacy disclosure and protection in OSNs in general. Li \emph{et al.} \cite{LiAll14} used the neighbors' locations to infer the location in the emerging location-based social networks. Sun \emph{et al.} \cite{SunSec15} protected the location privacy on the social crowdsourcing networks. Mao \emph{et al.} \cite{MaoLoo11} used the tweets to detect the Twitter users' situational leak such as vacation status, drunk status, and medical conditions. Dey \emph{et al.} \cite{DeyEst12} leveraged the information from neighbors to estimate the age of Facebook users. Mislove \emph{et al.} \cite{MisloYou10} also used the local connections around the Facebook users to infer their hidden attributes such as major, college, and political view. Our paper is complementary to these work and also highlights that current OSNs have emerged as an arguable threat to users' privacy.

\section{Conclusion and Future Work}\label{sec:CFW}
This paper presented LocInfer, a novel system that is able to discover the majority of Twitter users in any geographic area. Detailed experiments confirmed the high efficacy and efficiency of LocInfer. We also proposed a countermeasure to hide the locations of sensitive users from LocInfer and evaluated its efficacy with experiments driven by real datasets.

There are some open issues to study in our future work. First, when constructing a reliable seed set $S$ at the first step, we assumed the credibility of the self-reported locations and used the heuristic method to refine the seed set. For the future work, more advanced methods can be used to refine the seed set, and it is also interesting to check the impact of the seed set credibility on the ultimate performance. Second, the accuracy of LocInfer can be further improved by incorporating the existing content-based methods \cite{HechtTwe11, ChengYou10, MahmuWhe12, MahmuHom14} and other signals such as timezone and language. 

\section*{Acknowledgments}
We truly appreciate the anonymous reviewers for their constructive comments. This work is partially supported by ARO through W911NF-15-1-0328.

\bibliographystyle{IEEETran}

\begin{thebibliography}{22}
\bibitem{LeetaMap13}
K.~Leetaru, \emph{et al.}, ``Mapping the global
  {Twitter} heartbeat: The geography of {Twitter},'' \emph{First Monday},
  vol.~18, no.~5, 2013. [Online]. Available:
  \url{http://journals.uic.edu/ojs/index.php/fm/article/view/4366}.

\bibitem{LiTow12}
R.~Li, S.~Wang, H.~Deng, R.~Wang, and K.~Chang, ``Towards social user
  profiling: Unified and discriminative influence model for inferring home
  locations,'' in \emph{KDD'12}, Beijing, China, Aug. 2012, pp. 1023--1031.

\bibitem{HechtTwe11}
B.~Hecht, L.~Hong, B.~Suh, and E.~Chi, ``Tweets from {Justin Bieber}'s heart:
  the dynamics of the location field in user profiles,'' in \emph{CHI'11},
 Vancouver, Canada, May 2011, pp. 237--246.

\bibitem{ChengYou10}
Z.~Cheng, J.~Caverlee, and K.~Lee, ``You are where you tweet: a content-based
  approach to geo-locating {Twitter} users,'' in \emph{CIKM'10}, Toronto, Canada, Oct. 2010, pp. 759--768.

\bibitem{MahmuWhe12}
J.~Mahmud, J.~Nichols, and C.~Drews, ``Where is this tweet from? inferring home
  locations of {Twitter} users,'' in \emph{ICWSM'12}, Dublin, Ireland, Jun. 2012, pp. 511--514.

\bibitem{MahmuHom14}
J.~Mahmud, J.~Nichols, and C.~Drews, ``Home location identification of {Twitter} users,'' \emph{ACM Trans. Intell. Syst. Technol.}, vol.
  5, num. 3, Jul. 2014, pp. 47:1--47:21.

\bibitem{BacksFin10}
L.~Backstrom, E.~Sun, and C.~Marlow, ``Find me if you can: Improving
  geographical prediction with social and spatial proximity,'' in
  \emph{WWW'10}, Raleigh, NC, Apr. 2010, pp. 61--70.

\bibitem{McGeeLoc13}
J.~McGee, J.~Caverlee, and Z.~Cheng, ``Location prediction in social media
  based on tie strength,'' in \emph{CIKM '13}, San Francisco, CA, Oct. 2013, pp. 459--468.

\bibitem{JurgeTha13}
D.~Jurgens, ``That's what friends are for: Inferring location in online social
  media platforms based on social relationships,'' in \emph{ICWSM'13}, Boston, MA, Jul. 2013, pp. 273--282.

\bibitem{YamagLan13}
Y.~Yamaguchi, T.~Amagasa, and H.~Kitagawa, ``Landmark-based user location
  inference in social media,'' in \emph{COSN'13}, Boston, MA, Oct. 2013, pp. 223--234.

\bibitem{ComptGeo14}
R.~Compton, D.~Jurgens, and D.~Allen, ``Geotagging one hundred million {Twitter}
  accounts with total variation minimization,'' \emph{IEEE Big Data'14, } Oct. 2014, pp. 393-401.

\bibitem{QuercSoc12}
D.~Quercia, L.~Capra, and J.~Crowcroft, ``The social world of {Twitter}:
  Topics, geography, and emotions,'' in \emph{ICWSM'12}, Dublin, Ireland, Jun. 2012, pp. 298--305.

\bibitem{SridhTwi12}
V.~Sridharan, V.~Shankar, and M.~Gupta, ``Twitter games: How successful
  spammers pick targets,'' in \emph{ACSAC'12}, Orlando, FL, Dec. 2012, pp. 389--398.

\bibitem{api13}
{Twitter}, ``{REST API} v1.1 resources,'' 2013. [Online]. Available:
  \url{https://dev.twitter.com/docs/api/1.1/}.

\bibitem{gazetteer13}
``2013 {U.S.} gazetteer files.'' [Online]. Available:
  \url{http://www.census.gov/geo/maps-data/data/gazetteer2013.html}.

\bibitem{YangAna12}
C.~Yang, R.~Harkreader, J.~Zhang, S.~Shin, and G.~Gu, ``Analyzing spammers'
  social networks for fun and profit -- a case study of cyber criminal
  ecosystem on {Twitter},'' in \emph{WWW'12}, Lyon, France, Apr. 2012, pp. 71--80.

\bibitem{GhoshUnd12}
S.~Ghosh, B.~Viswanath, F.~Kooti, N.~Sharma, K.~Gautam, F.~Benevenuto, N.~Ganguly, and K.~Gummadi, ``Understanding and combating link farming in the
  {Twitter} social network,'' in \emph{WWW'12}, Lyon, France, Apr. 2012, pp. 61--70.

\bibitem{CLRS}
T.~Cormen, C.~Stein, R.~Rivest, and C.~Leiserson, \emph{Introduction to Algorithms, Third Edition}, Chapter 6, Cambridge MA: MIT Press, 2009, pp. 162-169.

\bibitem{emarketer14}
{eMarketer}, ``{US Twitter} user base begins to mature,'' Feb. 2014. [Online].
  Available:
  \url{http://www.emarketer.com/Article/US-Twitter-User-Base-Begins-Mature/1010641/2}.


\bibitem{LiAll14}
M.~Li, H.~Zhu, Z.~Gao, S.~Chen, L.~Yu, S.~Hu, and K.~Ren,`` All Your Location Are Belong to Us:
Breaking Mobile Social Networks for Automated User Location Tracking,'' in \emph{MobiHoc'14},
Philadelphia, PA, Aug. 2014, pp. 43--52.

\bibitem{SunSec15}
J.~Sun, R.~Zhang, X.~Jin, Y.~Zhang, ``SecureFind: Secure and Privacy-Preserving Object Finding
 via Mobile Crowdsourcing,'' arXiv:1503.07932, http://arxiv.org/pdf/1503.07932.pdf, 2015.

\bibitem{MaoLoo11}
H.~Mao, X.~Shuai, and A.~Kapadia, ``Loose tweets: An analysis of privacy leaks
  on {Twitter},'' in \emph{WPES'11}, Chicago, IL, Oct. 2011, pp. 1--12.

\bibitem{DeyEst12}
R.~Dey, T.~Cong, K.~Ross, and N.~Saxena, ``Estimating age privacy leakage in
  online social networks,'' in \emph{INFOCOM'12}, Orlando, FL, Mar. 2012, pp. 2836--2840.

\bibitem{MisloYou10}
A.~Mislove, B.~Viswanath, K.~Gummadi, and P.~Druschel, ``You are who you know:
  inferring user profiles in online social networks,'' in \emph{WSDM'10}, New York City, NY, Feb. 2010, pp. 251--260.



\end{thebibliography}

\end{document}